\documentclass[prx,reprint,floatfix,letterpaper,twocolumn,nofootinbib,showpacs,longbibliography]{revtex4-1}
\usepackage[caption=false]{subfig}
\usepackage{graphicx} 
\usepackage{epstopdf}
\usepackage{array}
\usepackage{verbatim}
\usepackage{amsmath,amsfonts,amssymb,amscd}
\usepackage{amsthm}
\usepackage{tabularx}
\usepackage{stmaryrd}
\usepackage{enumerate}
\usepackage[ruled]{algorithm2e}
\usepackage{enumitem}
\usepackage{wasysym}

\usepackage{hyperref}
\hypersetup{
    bookmarksnumbered=true, 
    unicode=false, 
    pdfstartview={FitH}, 
    pdftitle={}, 
    pdfauthor={}, 
    pdfsubject={}, 
    pdfcreator={}, 
    pdfproducer={}, 
    pdfkeywords={}, 
    pdfnewwindow=true, 
    colorlinks=true, 
    linkcolor=blue, 
    citecolor=blue, 
    filecolor=blue, 
    urlcolor=blue 
}

\theoremstyle{plain}
\newtheorem{thm}{Theorem}

\newtheorem{lem}[thm]{Lemma}
\newtheorem{claim}[thm]{Claim}

\theoremstyle{definition}
\newtheorem{defn}[thm]{Definition}

\newcommand{\eq}[1]{(\hyperref[eq:#1]{\ref*{eq:#1}})}

\renewcommand{\sec}[1]{\hyperref[sec:#1]{Section~\ref*{sec:#1}}}
\newcommand{\thrm}[1]{\hyperref[thm:#1]{Theorem~\ref*{thm:#1}}}
\newcommand{\lemm}[1]{\hyperref[lemm:#1]{Lemma~\ref*{lemm:#1}}}
\newcommand{\prop}[1]{\hyperref[prop:#1]{Proposition~\ref*{prop:#1}}}
\newcommand{\corr}[1]{\hyperref[corr:#1]{Corollary~\ref*{corr:#1}}}
\newcommand{\fig}[1]{\hyperref[fig:#1]{Figure~\ref*{fig:#1}}}

\newcommand{\ket}[1]{|#1\rangle}
\newcommand{\bra}[1]{\langle#1|}

\DeclareMathAlphabet{\matheu}{U}{eus}{m}{n}

\newcolumntype{L}[1]{>{\raggedright}p{#1}}
\newcolumntype{C}[1]{>{\centering}p{#1}}
\newcolumntype{R}[1]{>{\raggedleft}p{#1}}
\newcolumntype{D}{>{\centering\arraybackslash}X}

\begin{document}
\title{Universal fault-tolerant quantum computation with Bacon-Shor codes}
\author{Theodore J. Yoder}
\affiliation{Department of Physics, Massachusetts Institute of Technology}
\begin{abstract}
We present a fault-tolerant universal gate set consisting of Hadamard and controlled-controlled-Z (CCZ) on Bacon-Shor subsystem codes. Transversal non-Clifford gates on these codes are intriguing in that higher levels of the Clifford hierarchy become accessible as the code becomes more asymmetric. For instance, in an appropriate gauge, Bacon-Shor codes on an $m\times m^k$ lattice have transversal $k$-qubit-controlled $Z$. Through a variety of tricks, including intermediate error-correction and non-Pauli recovery, we reduce the overhead required for fault-tolerant CCZ. We calculate pseudothresholds for our universal gate set on the smallest $3\times3$ Bacon-Shor code and also compare our gates with magic-states within the framework of a proposed ion trap architecture.
\end{abstract}

\maketitle
\parskip 0pt
Shor's 9-qubit code was the first quantum code discovered \cite{Shor1995} and is still popular due to its conceptual simplicity. Later, it was realized that viewing Shor's code as a subsystem code (the so-called Bacon-Shor code) leads to even easier protocols for error-correction \cite{Bacon2006,Aliferis2007} with just local interactions on an $m\times n$ lattice of qubits. An important concept in quantum coding theory is a threshold, the value of the physical error rate below which encoding quantum data begins to help reduce errors. Suprisingly, the simple Bacon-Shor code even boasts some of the highest known thresholds for concatenated codes \cite{Aliferis2007}.

Behind such thresholds, however, actually lie several assumptions on how a universal gate set is constructed \cite{Aliferis2007b}. Since the quoted thresholds are calculated by simulating the encoded CNOT gadget, but are for \emph{universal} computation, the assumption is made that for non-Clifford gates, a magic-state \cite{Bravyi2005} should be distilled and injected at an arbitrarily high level of concatenation. The resulting scheme is simply not resource realistic for fault-tolerant experiments operating in the low-distance limit.

It is more realistic to study a universal gate set for low-distance codes directly, and it is this program that we adhere to here. This is possible in principle, as already evidenced dating back to Shor \cite{Shor1996} and similar constructions elsewhere \cite{Aliferis2006} that verify magic-states directly at low distance. Optimizing low-distance constructions of non-Clifford gates is the next step. Prior work in this direction is promising. For instance, gauge-fixing can convert between the 7-qubit and 15-qubit codes \cite{Paetznick2013,Anderson2014} to take advantage of their complementary universal transversal gate sets. Alternatively, concatenation of complementary codes also yields a universal set of gates \cite{Jochym2014}.

Our starting point is a different strategy, wherein instead of combining two different codes, non-transversal gates are constructed directly and made fault-tolerant via stabilizer measurements intermediate in the circuit. Since the intermediate error-correction cycles effectively break the circuit into fault-tolerant pieces, this is dubbed ``pieceable'' fault-tolerance \cite{Yoder2016}. One advantage of this approach is relatively broad applicability, including to the Bacon-Shor code family. Code-specific simplifications of the general-purpose designs in \cite{Yoder2016} exist and these can offer substantial improvements in how much intermediate error-correction is required. The Bacon-Shor code's simple structure also lends itself well to these strategies.

Here we develop an appealing scheme for universal fault-tolerant computing on Bacon-Shor codes using fault-tolerant Hadamard ($H$) and controlled-controlled-Z (CCZ) gates. In its simplest form, the scheme hinges on the observation that extending a symmetric Bacon-Shor code into an asymmetric one enlarges the class of transversal non-Clifford gates on the code. Indeed, make the code asymmetric enough, i.e.~$m\times m^k$, and a gate (namely the $k$-qubit-controlled Z) from the $k^{\text{th}}$-level of the Clifford hierarchy \cite{Gottesman1999} becomes transversal. If the code does not meet the asymmetry requirement, then adding intermediate error-correction can still make the construction fault-tolerant. Since many uses of intermediate error-correction are undesirable, we can reduce the number by clever circuit design. Additional simplifications are achieved using non-Pauli recovery operations. Non-Pauli recovery lies outside the standard formalism for fault-tolerant computing with stabilizer codes but is permitted by the Knill-Laflamme conditions \cite{Knill1997}. Such overhead-reducing innovations may be crucial for experimentally realizing small instances of fault-tolerance.

The smallest instance of our constructions is a fault-tolerant CCZ on the $3\times3$ Bacon-Shor code that uses no intermediate error-correction. Because $H$ is transversal on symmetric Bacon-Shor codes, the $3\times3$ code has a particularly simple fault-tolerant universal gate set that moreover requires no postselection.

For the $3\times3$ code, we compute exREC pseudothresholds \cite{Svore2006} for a universal gate set under circuit depolarizing noise and find $\sim8\times10^{-5}$ for the largest gate CCZ. We also discuss overhead, and compare with magic-state implementations of CCZ in the ion trap MUSICQ \cite{Monroe2014} architecture. In this context, we estimate a roughly 4 times faster implementation of fault-tolerant Toffoli, a key part of the quantum circuits used for Shor's factoring \cite{Shor1999}. We conclude that ion traps, with their ability to easily implement non-local gates, are in a promising position to take full advantage of fault-tolerant optimizations.

The Bacon-Shor codes are subsystem codes, as opposed to subspace codes. As such, there are three relevant sets of Paulis, the stabilizers, gauge operators, and logical operators. Let $[m)=\{0,1,\dots,m-1\}$ and $(m)=\{1,2,\dots,m-1\}$ denote subsets of integers. Lay out the code qubits in an $m\times n$ lattice (Fig.~\ref{Bacon_shor_grid}), so qubits are indicated by coordinates $(i,j)\in[m)\times[n)$ and single-qubit Paulis by subscripts $X_{i,j},Y_{i,j}$, and $Z_{i,j}$. As in \cite{Aliferis2007}, denote columns and rows of Paulis by, respectively,
\begin{align}
Z_{*,j}&=Z_{0,j}Z_{1,j}\dots Z_{m-1,j},\\
X_{i,*}&=X_{i,0}X_{i,1}\dots X_{i,n-1}.
\end{align}

\begin{figure}
\includegraphics[width=0.8\columnwidth]{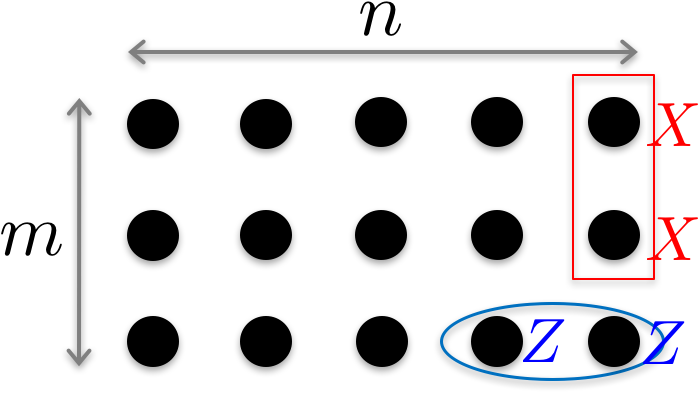}
\caption{\label{Bacon_shor_grid} The $m\times n$ Bacon-Shor code on a lattice with $Z$-type and $X$-type gauge operators labeled. All gauge operators are translates of these, and all stabilizers are products of them.}
\end{figure}

The stabilizer group $S$ is generated from
\begin{align}
\tilde Z_j&=Z_{*,j-1}Z_{*,j},\quad\forall j\in(n),\\
\tilde X_i&=X_{i-1,*}X_{i,*},\quad\forall i\in(m).
\end{align}
That is, the stabilizers are an even number of columns of $Z$s or an even number of rows of $X$s.

Gauge operators are generated by
\begin{align}
\bar Z_{i,j}&=Z_{i,j-1}Z_{i,j},\quad\forall i\in[m),j\in(n),\\
\bar X_{h,k}&=X_{h-1,k}X_{h,k},\quad\forall h\in(m),k\in[n),
\end{align}
while (lowest-weight) bare logical operators are any single column of $Z$s and any single row of $X$s. For instance, $\bar Z=Z_{*,0}$ and $\bar X=X_{0,*}$. The code distance of an $m\times n$ Bacon-Shor code is $\min(m,n)$, but asymmetric codes also offer greater protection against one type of error \cite{Brooks2013}. A (destructive) transversal measurement of all qubits in the $Z$- ($X$-) basis suffices to measure $\bar X$ and $\bar Z$.

The gauge offers degrees of freedom not available in subspace stabilizer codes. Indeed, notice that from measurements of the gauge operators we can infer a measurement of the stabilizers, since
\begin{align}
\tilde Z_j = \prod_{i\in[m)} \bar Z_{i,j},\quad \tilde X_i = \prod_{j\in[n)} \bar X_{i,j}.
\end{align}
Measuring local gauge operators can be easier than measuring non-local, high-weight stabilizers directly.

Although the gauge offers an advantage during error-correction, during our logical gates we will want to fix a gauge. On paper, fixing a gauge amounts to adding a maximal, commuting subset of the gauge operators to the stabilizer. Common gauges are the $Z$-gauge (formed by placing all $\bar Z_{i,j}$ into $S$) and the $X$-gauge (placing all $\bar X_{h,k}$ into $S$). However, the rotated surface code \cite{Tomita2014} is also just a gauge choice of the Bacon-Shor code.

In practice, we can fix a gauge by measuring the gauge operators. Steane error-correction \cite{Steane1997} achieves a high threshold while also being extremely simple for Bacon-Shor codes. Indeed, the logical states $\ket{\bar 0_X}$ and $\ket{\bar +_Z}$, where subscripts denote $X$- or $Z$-gauge, are simply tensor products of CAT states, $\ket{\bar0_X}=\left(|+\rangle^{\otimes m}+|-\rangle^{\otimes m}\right)^{\otimes n}$ and $\ket{\bar+_Z}=\left(|0\rangle^{\otimes n}+|1\rangle^{\otimes n}\right)^{\otimes m}$.
Preparing these fault-tolerantly is easy for $m=n=3$, because 3-qubit CAT states need not be postselectively verified. Gauge operators $\bar X_{h,k}$ (and therefore stabilizers $\tilde X_i$ as well) are measured using $\ket{\bar0_X}$, while $\bar Z_{i,j}$ (and so $\tilde Z_j$ as well) are measured using $\ket{\bar+_Z}$. Measuring the $X$-gauge followed by the $Z$-gauge leaves the code in the $Z$-gauge, and reversing the order of measurement leaves it in the $X$-gauge.

There is a subtlety in the process of gauge fixing. When the syndrome measurement is ordered such that the gauge changes (from $X$- to $Z$- or vice versa) we obtain more information about the errors than if we had ordered the measurements such that the gauge does not change. In the former case, call it type-1 correction, if the code began in the $Z$-gauge ($X$-gauge), we learn the values of all $\bar Z_{i,j}$ ($\bar X_{h,k}$) \emph{and} the values of all $\tilde X_i$ ($\tilde Z_j$). However, in the latter case, type-2 correction, we learn only the values of all $\tilde X_i$ and all $\tilde Z_j$. Transversal logical gates can use either type of correction to achieve fault-tolerance, but the fault-tolerance of our non-transversal CCZ gates can be dependent on the added information gathered from type-1 error-correction. Nevertheless, if the gauge change of type-1 is undesired, it can always be changed back by a subsequent type-1 correction.

Extending a Bacon-Shor code can be done fault-tolerantly. To transfer the encoded quantum state from an $m\times n$ code to an $m'\times n$ code, $m'>m$, prepare an $(m'-m)\times n$ Bacon-Shor codeblock in $\ket{\bar+_Z}$. Join the ancilla block with the initial block by measuring the $Z$-gauge operators across the boundary. This could be done by performing $m'\times n$ Steane error-correction for instance. Adding more rows instead is the Hadamard conjugate of this process. One can also remove columns by measuring the individual qubits in the $X$-basis and rows by measuring them in the $Z$-basis.

It is well known that Toffoli and Hadamard are a universal set of gates for quantum computation \cite{Aharonov2003}. Replacing Toffoli with controlled-controlled-$Z$ (CCZ) also makes a universal gate set. We now discuss how to implement Hadamard and CCZ on a Bacon-Shor code.

On symmetric (i.e.~$m\times m$) Bacon-Shor codes, logical Hadamard $\overline{H}$ is a transversal gate up to a qubit permutation \cite{Aliferis2007}. 
On asymmetric Bacon-Shor codes (i.e.~$m\times n$ with $n>m$), $\overline{H}$ can be done via teleportation \cite{Zhou2000}. Preparing $\ket{\overline{+}_Z}$, coupling to the target codeblock with $\overline{\text{CZ}}$ (which we show later is transversal), and measuring $\bar X$ on the target codeblock suffices to teleport the original encoded state to the ancilla codeblock with $\overline{H}$ applied. 
This protocol also implies CCZ is universal on its own (given $X$- and $Z$-basis state preparation and measurement). We discuss this corollary more in Appendix~\ref{no_1qubit_univ}.


To obtain computational universality, we implement logical CCZ, a three-qubit gate, with full-distance (i.e.~under circuit depolarizing noise, a distance $d$ code recovers from $\lfloor(d-1)/2\rfloor$ faulty circuit components). Assume that all code blocks (labeled $A$, $B$, $C$) begin in the $Z$-gauge. Logical CCZ, denoted $\overline{\text{CCZ}}$, can be created from physical CCZ gates in round-robin fashion \cite{Yoder2016}:
\begin{equation}\label{rr_ccz}
\overline{\text{CCZ}}=\prod_{u,v,w\in\text{supp}(\bar Z)}\text{CCZ}(u_A,v_B,w_C).
\end{equation}
All gates following $\prod$-symbols in this paper mutually commute, so ordering is unnecessary.

Ostensibly then, $m\times n$ Bacon-Shor codes would use $m^3$ physical CCZ gates to implement $\overline{\text{CCZ}}$ in a depth $m^2$ circuit, because $\bar Z$ has support $\text{supp}(\bar Z)$ of size at least $m$. To make Eq.~\eqref{rr_ccz} fault-tolerant, it is sufficient to measure all of $\bar Z_{i,j}$ after each timestep of CCZ gates. This suffices because all $X$ errors can be detected and corrected before they propagate $Z$ errors through subsequent CCZ gates. In the terminology of \cite{Yoder2016}, we say the circuit is fault-tolerant in $m^2$ pieces, a number of pieces equal to the circuit depth. However, two kinds of simplifications can generally reduce the number of pieces.

The first of these simplifications exploits the code stabilizer to reduce the depth of the circuit. By definition, for a stabilizer $s\in S$ and a state $\ket{\bar\psi}$ in the codespace, $s\ket{\bar\psi}=\ket{\bar\psi}$. Thus, we might say $s$ is an implementation of logical identity, $\bar I$. However, it is only one such implementation. For instance, controlled-$s$ and controlled-controlled-$s$ are also implementations of $\bar I$ for any control qubit(s) (note, the controls could be taken in any basis, but we will only use controls in the $Z$-basis here). In $Z$-gauge Bacon-Shor codes, $\bar Z_{i,j}\in S$ and thus,
\begin{equation}\label{cci}
\bar I = \text{CCZ}(u_A,v_B,(i,j-1)_C)\text{CCZ}(u_A,v_B,(i,j)_C),
\end{equation}
for all $u\in[m)\times[n)$, $v\in[m)\times[n)$, $i\in[m)$, and $j\in(n)$. Similar expressions hold under permutation of $A$, $B$, $C$.

By multiplying $\overline{\text{CCZ}}$ from Eq.~\ref{rr_ccz} by implementations of $\bar I$ from Eq.~\ref{cci}, we create lower depth implementations of $\overline{\text{CCZ}}$. This can be visualized as moving the control nodes of the CCZ gates across rows of a $Z$-gauge Bacon-Shor code block, thereby spreading CCZ gates across all code qubits. For $m\times m$ Bacon-Shor codes, it is trivial to reduce the depth of the circuit to $m$. For $m\times m^2$ Bacon-Shor codes, $\overline{\text{CCZ}}$ becomes depth-1, i.e.~transversal. Fig.~\ref{BS39ccz} shows the $3\times9$ case. Generally, $m\times n$ Bacon-Shor codes can implement $\overline{\text{CCZ}}$ with a depth $\lceil m^2/n\rceil$ circuit, which translates directly to a fault-tolerant $\overline{\text{CCZ}}$ in $\lceil m^2/n\rceil$ pieces, as discussed above. From this argument we see a space-time tradeoff emerge: an $m\times n$ Bacon-Shor code supports a depth $h$ circuit for $\overline{\text{CCZ}}$ if $hn\ge m^2$. The appropriate generalization for a $\overline{\text{C}^k\text{Z}}$ gate (a $Z$ with $k$ controls) is $hn\ge m^k$ for any integer $k\ge1$.

\begin{figure}
\includegraphics[width=\columnwidth]{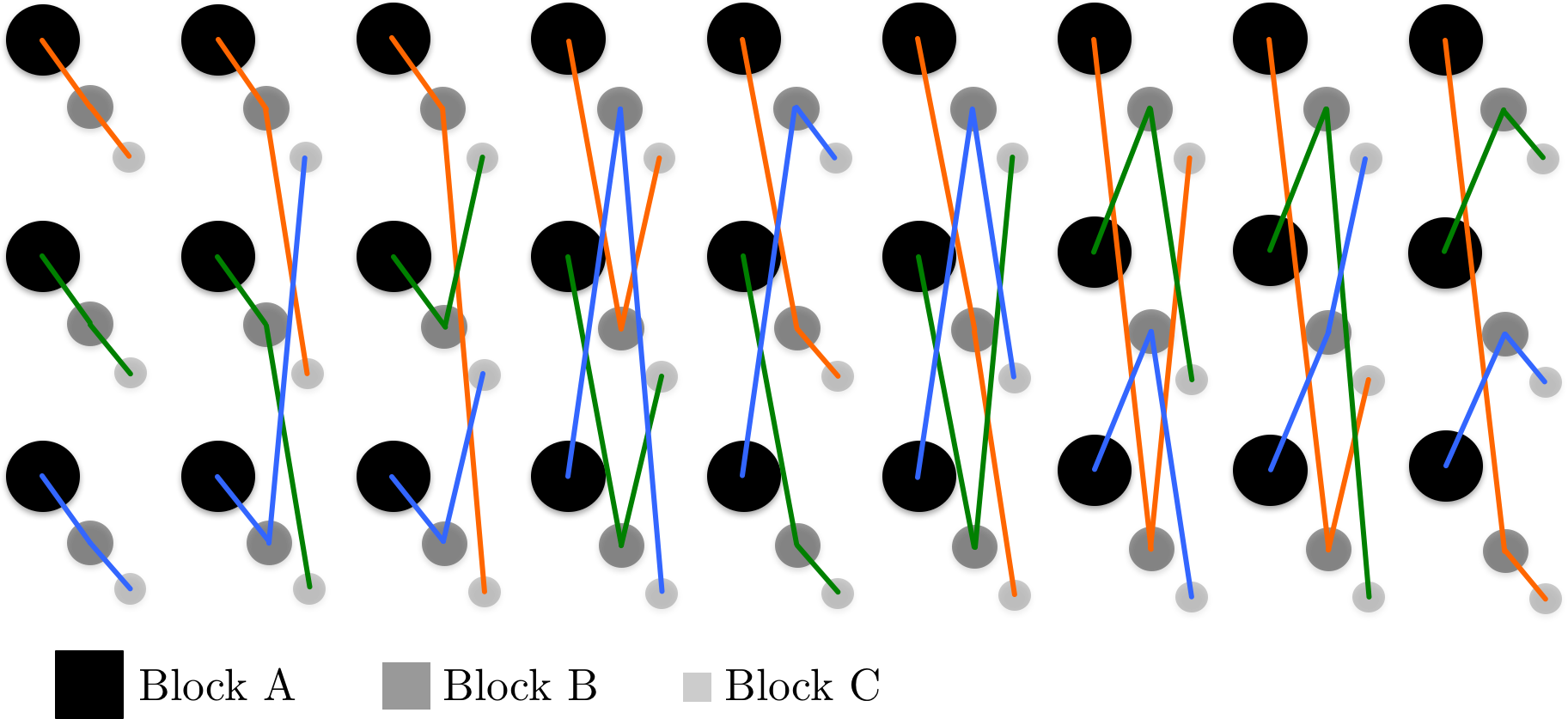}
\caption{\label{BS39ccz} Visualizing transversal $\overline{\text{CCZ}}$ on the $3\times9$ Bacon-Shor code. Orange, green, and blue jointed lines indicate physical CCZ gates between qubits of the three codeblocks.}
\end{figure}

Explicitly, these circuits for $\overline{\text{C}^k\text{Z}}$ on $m\times n$ codes can be arranged in terms of subcircuits indexed by a $k$-digit $m$-ary number $p$. The subcircuit $p=p_{k-1}p_{k-2}\dots p_0$ is
\begin{equation}
C_p=\prod_{i=0}^{m-1}\text{C}^{k}\text{Z}\left((i,j),(i\oplus p_{k-1},j),\dots(i\oplus p_0,j)\right)
\end{equation}
with addition $\oplus$ modulo $m$ and for some choice of column $j$ depending on $p$. The standard choice would be $j=p\text{ }(\text{mod } n)$. Subcircuits with the same value of $j$ must be done in subsequent timesteps, contributing to the circuit depth $h=\lceil m^k/n\rceil$. The product of all subcircuits implements the logical gate: $\overline{\text{C}^k\text{Z}}=\prod_{p=0}^{m^k-1}C_p$. The advantage of this organization is that, if the three interacting codeblocks are layered upon one another in the plane, the physical $\text{C}^k\text{Z}$ gates interact qubits within a column, i.e.~distanced from one another by at most $m-1$ lattice spacings (rather than the worst case $n-1\ge m-1$). 

That $\overline{\text{C}^k\text{Z}}$ with $k>1$ on 2D Bacon-Shor codes must use long-range gates (or, equivalently, SWAP circuits of non-constant depth) is necessitated by arguments similar to those of Bravyi-K\"{o}nig \cite{Bravyi2013} for topological subspace codes. In Appendix~\ref{BK-BS}, we present this argument and note that for all $k$ our constructions use optimal  gate range.

As a practical matter, either substantially extending the code or using a larger depth circuit and correcting $X$ errors after every timestep may be unappealing. As our second simplification, we can reduce the number of intermediate error corrections by using an idea from \cite{Yoder2016} called 2-transversality, wherein each qubit interacts with at most two qubits from each other codeblock. We leave a more thorough description of this simplification to Appendix~\ref{pauli_cond}, but note here that it can reduce the number of pieces used to implement $\overline{\text{C}^k\text{Z}}$ to $\lceil\lceil m/2\rceil^k/n\rceil$. In particular, a $3\times4$ Bacon-Shor code (just 12 code qubits) can implement $\overline{\text{CCZ}}$ without intermediate error-correction.

However, we can do even better, implementing $\overline{\text{CCZ}}$ on smaller Bacon-Shor codes without the need for intermediate error-correction. In doing so, we no longer assume that error-correction consists of projection to the stabilizer space with strictly Pauli recovery. This highlights the inequivalence of the Knill-Laflamme conditions for the existence of a general recovery map and the (necessarily stronger) conditions for the existence of Pauli recovery. See Appendix~\ref{pauli_cond} for these conditions.

Using this more general error-correction, our smallest $\overline{\text{CCZ}}$ construction is built upon the $3\times3$ Bacon-Shor code in the $Z$-gauge. The logical gate takes three timesteps of CCZ gates followed by error-correction. Explicitly, the CCZ circuit at timestep $t\in\{0,1,2\}$ is
\begin{equation}\label{CCZ33}
\prod_{i=0}^{2}\prod_{j=0}^{2}\text{CCZ}\big((i+f(j,t),j)_A,(i,j)_B,(i+g(j,t),j)_C\big)
\end{equation}
for $f(j,t)=j+\lfloor t/2\rfloor$, $g(j,t)=-j+\lceil t/2\rceil$. The error-correction measures the $Z$-gauge operators, applies Pauli $X$ and CZ corrections, then measures the $X$-gauge operators and applies Pauli $Z$ corrections. 
We argue this is fault-tolerant at the end of Appendix~\ref{pauli_cond}.

To illustrate the advantages of such a small $\overline{\text{CCZ}}$ construction, we provide a comparison with a magic state method for ion traps. In the MUSICQ architecture \cite{Monroe2014}, qubits are grouped into elementary logical units, or ELUs, with $N_q$ qubits per ELU (where $N_q\le100$ is considered daunting but possible) all of which may interact via two-qubit gates. Qubits within different ELUs interact by teleportation through shared entanglement generated by photon interference. Since it takes roughly two orders of magnitude longer to generate entanglement between ELUs than it takes to interact qubits within an ELU, we set aside $C_q\approx N_q/2$ qubits per ELU for interaction with the other ELUs. Due to excellent state lifetimes in ion traps, it is reasonable to assume that entanglement is generated and stored just before a computation. 

Thus, the MUSICQ architecture excels at implementing non-local gates. This comes at the cost of limited, though not nonexistent, parallel operations. To compare our Bacon-Shor $\overline{\text{CCZ}}$ with magic-states for Steane's 7-qubit code, we follow \cite{Monroe2014} and assume up to twelve multi-qubit operations (CNOTs and CCZs) can be performed in parallel within an ELU. We idealize single-qubit gate and state preparation time as 1$\mu s$, 2-qubit and 3-qubit gate time  as 10$\mu s$, and measurement time as 30$\mu s$ \cite{Monroe2014}.

\setlength\extrarowheight{4pt}
\begin{table}
\begin{tabular}{|c|c|c|c|c|}
\hline
& Circ.~Vol. & Spacetime & Time & Qubits\\
\hline\hline
Magic 7 & 1,400 & 19,900 $\mu s\times\text{qub.}$ & 940$\mu s$ & 66\\
\hline
Magic 9 & 1,100 & 15,800 $\mu s\times\text{qub.}$ & 910$\mu s$ & 81\\
\hline
BS $3\times3$ & 440 & 5,540 $\mu s\times\text{qub.}$ & 190$\mu s$ & 54\\
\hline
\end{tabular}
\caption{\label{tab1} Rough comparison of the magic-state preparation and injection protocol on the 7-qubit code considered in \cite{Monroe2014}, the analogous protocol for the $3\times3$ code, and our optimized $3\times3$ construction for implementing $\overline{\text{CCZ}}$. Circuit volume counts circuit components weighted by the number of qubits involved, while spacetime volume does the same but also weighted by the time of physical implementation. By symmetry, numbers for logical Toffoli are identical.}
\end{table}

\begin{table}
\begin{tabular}{|c|c|c|}
\hline
Gate & $p$ & $p/10^{2-q}$\\
\hline\hline
$I$ \& $H$ & $4.1\times10^{-4}$ & $1.9\times10^{-4}$\\
\hline
CNOT & $1.4\times10^{-4}$ & $5.3\times10^{-4}$\\
\hline
CCZ & $8.2\times10^{-5}$ & $6.1\times10^{-4}$\\
\hline
\end{tabular}
\caption{\label{tab2} A table of exREC \cite{Aliferis2006} pseudothreshold \cite{Svore2006} lower bounds for gates on the $3\times3$ Bacon-Shor code. The thresholds of all gates are calculated by exact counting. Shown are two versions of the circuit depolarizing noise model, one in which components all err with the same probability $p$ and the other in which a $q$-qubit component errs with probability $p/10^{2-q}$. The thresholds shown are the error rates for CNOT below which the encoded gate is better than physical.
}
\end{table}

Our comparison is shown in Table~\ref{tab1}. Roughly we expect a Bacon-Shor $\overline{\text{CCZ}}$ to be 4-times faster than using a magic state while also using fewer qubits. Also note that in all scenarios, the magic-state approach uses postselection to both prepare the magic-state and prepare CAT or Steane states for error-correction. In contrast, the Bacon-Shor $\overline{\text{CCZ}}$ never uses postselection. Magic-state $\overline{\text{CCZ}}$ for the 7-qubit code, following the design in \cite{Zhou2000,Monroe2014}, uses 56-126 qubits depending on error-correction scheme, with 66 striking a balance between expected error rate, qubit count, and gate time. Our Bacon-Shor CCZ can use between 30-81 qubits, with 54 ($9\times3$ data plus $9\times3$ ancillas) striking a good balance. Notice that the Bacon-Shor scheme leaves a comfortable 46 qubits per ELU for entanglement generation. Physical qubits are reusable in our estimations and we assume all qubits are in the same ELU. If not, teleporting codeblocks to the same ``interacting'' ELU adds a small overhead to the numbers in Table~\ref{tab1}.
More circuit details are in Appendix~\ref{gate_detail}.

We have also calculated pseudothresholds for our universal gate set on the $3\times3$ Bacon-Shor code using exact counting. The results are shown in Table~\ref{tab2}. Calculation details and logical error rate plots are in Appendix~\ref{thr_detail}.

What is the most practical route toward scalable, universal, fault-tolerant quantum computation? The question is fraught with many dependencies and subtleties, such as hardware capabilities, the gate set implemented, the noise model and rate, etc. If we drop the scalability requirement and set our sights instead on near term fault-tolerance at low distance, we can start optimizing.

It is in this spirit that we have developed our low-overhead universal computing scheme for the $3\times 3$ Bacon-Shor. And the results are relatively promising --- a high pseudothreshold and a size comfortably fitting into quantum computing architectures of the near future.

However, can we put scalability back without too much cost? Indeed, with the Bacon-Shor codes, this seems challenging. It is well-known that Bacon-Shor codes fail to have an asymptotic threshold as a topological family. Concatenation offers a threshold in theory, but in practice it is hard to implement and pays a price in overhead.

An intriguing alternative to improve scalability is to exploit the limited ability of 3D Bacon-Shor codes to be self-correcting \cite{Bacon2006}. $Z$-gauge Bacon-Shor codes in 3D are codespace equivalent to asymmetric codes in 2D. However, what were 1-dimensional rows of qubits in 2D codes become 2-dimensional lattices in 3D. The Ising $ZZ$ interaction can be applied between all lattice neighbors. Thus, the $Z$-gauge 3D Bacon-Shor code inherits the thermal (though not Hamiltonian-perturbative \cite{Pastawski2010}) stability of the 2D Ising model to protect $\bar X$. Moreover, since $X$ errors are suppressed in this manner, the need for intermediate correction in our circuits is reduced. Though not completely solving the scalability problem ($\bar Z$ is still vulnerable) it could offer a substantial simplification of our construction for architectures that can support it.

The author gratefully acknowledges Ken Brown, Steve Flammia, and Ryuji Takagi for discussions and their comments on the manuscript, as well as Cody Jones for ideas relating Bacon-Shor codes and the surface code. This project was supported financially by the Department of Defense (DoD) through the National Defense Science and Engineering Graduate (NDSEG) Fellowship program.

\bibliography{References.bib}

\appendix
\section{Computational universality without single-qubit gates}\label{no_1qubit_univ}
Most universal gate sets for quantum computation include some single-qubit gates. Indeed, this can be an efficient choice for implementing some algorithms because fault-tolerant versions of those gates are likely small, and, moreover, much more is known about compiling with single-qubit gates. However, single-qubit gates are not a necessary part of a computationally universal gate set, while multi-qubit gates are.

\begin{thm}
Assuming the availability of $\ket{0},\ket{+}$ and measurement in the $X$- and $Z$-bases, CCZ is quantum computationally universal.
\end{thm}
\begin{proof}
We start from the computationally universal set $\{H,CCZ\}$ \cite{Aharonov2003}, then implement $H$ using the circuit described in the main text (also see Fig.~\ref{tele_h}). This involves two gates not explicitly in our gate set. First, CZ can be implemented using CCZ controlled on a $\ket{1}$ ancilla. We can prepare $\ket{1}$ nondeterministically by measuring $\ket{+}$ in the $Z$-basis, succeeding with probability $1/2$. Second, $X$ might have to be applied via classical control. However, we can track this Pauli through the subsequent circuit. At most, it results in needing an additional CZ gate for every CCZ, a polynomial blowup in circuit size. We note that both complications are more easily remedied, in particular without any nondeterminism, by adding $X$ to the gate set, which is transversal for any stabilizer code.
\end{proof}

\begin{figure}
\includegraphics[width=0.9\columnwidth]{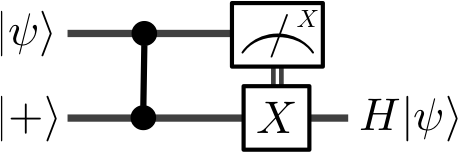}
\caption{\label{tele_h} Implementing $H$ by one-bit teleportation \cite{Zhou2000}.}
\end{figure}

Thus, any stabilizer code that can implement CCZ has fault-tolerant computational universality. Codes with transversal CCZ do not violate the Eastin-Knill no-go \cite{Eastin2009} because the injection of $H$ is not completely unitary. Interestingly, if transversal CZ is available in a CSS (Calderbank-Shor-Steane) stabilizer code, the injection of $H$ in Fig.~\ref{tele_h} uses no more resources than Steane error-correction would, because $Z$ errors are correctable using information from the transversal $X$-measurement. This is perhaps an even easier route to universality with the 15-qubit Reed-Muller code (and the larger quantum Reed-Muller family) than gauge-fixing \'{a} la \cite{Paetznick2013,Anderson2014}.

Finally, we note Toffoli is quantum universal on its own given the same access to preparation and measurement. Since classical reversible computation is universal with just Toffoli, this reveals the power of quantum computation as ``simply'' the ability to prepare and measure in the conjugate basis.

\section{Bravyi-K\"{o}nig for Bacon-Shor codes}\label{BK-BS}
In this section, we view the 2D Bacon-Shor codes as a topological family to better understand the limitations of logical gates. Bravyi-K\"{o}nig \cite{Bravyi2013} made general arguments restricting the ability of logical operators for stabilizer codes with local generators in $D$ spatial dimensions. Subsequently, Pastawski-Yoshida \cite{Pastawski2015} made similar arguments for subsystem codes that possess a threshold. Unfortunately, neither of these results directly apply to the Bacon-Shor CCZ gates that we have developed here. Our gates work only in the fixed $Z$-gauge, and therefore the $X$-stabilizers fail to satisfy the spatial locality constraint of Bravyi-K\"{o}nig. The Bacon-Shor code family also notoriously fails to have a threshold, and so fails to satisfy the assumptions of Pastawski-Yoshida.

Nevertheless, following the simpler argument of Bravyi-K\"{o}nig for the restriction of logical gates on the 2D surface code, we can develop the following theorem for 2D, $Z$-gauge Bacon-Shor codes
\begin{thm}
Consider a constant-depth circuit $U$ that is a logical operator on a constant number of copies a 2D, $Z$-gauge Bacon-Shor code with distance $d$. Then, $U$ is a Clifford operation as long as the gates in $U$ have ranges in the $x$- and $y$-dimensions $R_x$ and $R_y$, respectively, satisfying $(R_x+1)(R_y+1)< O(d)$.
\end{thm}
We note that our CCZ constructions on $m\times m^2$ codes saturate the bound, because $R_y=m-1=d-1$ and $R_x=0$. Indeed, even our $\text{C}^k\text{Z}$ constructions on $m\times m^k$ codes saturate the bound for all $k$.
\begin{proof}
Consider a stack of a constant number of 2D, $Z$-gauge Bacon-Shor codes depicted in Fig.~\ref{bacon_topo}. Denote pairs of horizontal and vertical regions of qubits by $\xi_j$ and $\zeta_j$ with $j=1,2$. These regions can be chosen to have constant width and be separated by $O(d)$ qubits. Consider two logical Pauli operators $P$ and $Q$ acting on all codeblocks. Without loss of generality we may take $P$ to lie within $\alpha_1:=\xi_1\cup\zeta_1$ and $Q$ to lie within $\alpha_2:=\xi_2\cup\zeta_2$.

\begin{figure}
\includegraphics[width=0.75\columnwidth]{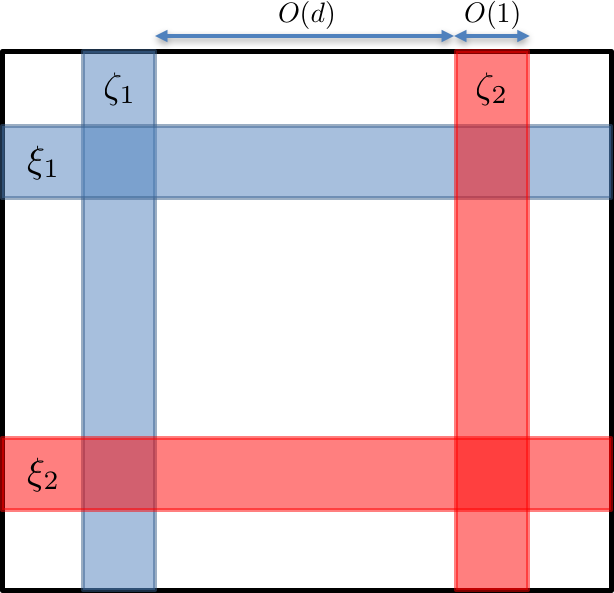}
\caption{\label{bacon_topo} A stack of a constant number (e.g.~3) of 2D, $Z$-gauge Bacon-Shor codes can be visualized in 2D.}
\end{figure}

Write the group commutator
\begin{equation}
K=P(UQU^\dag)P(UQU^\dag).
\end{equation}
It is clear that $K$ is a logical operator since $U$ is. Assume for now (to be shown later) that $K$ is a trivial logical operator, i.e.~
\begin{equation}\label{K_is_logical}
K\Pi=c\Pi,
\end{equation}
where $\Pi$ is the projector onto the codespace and $c$ is a constant. Writing $PK\Pi=cP\Pi$ and squaring it, we get
\begin{equation}
(PK)^2\Pi=c^2P^2\Pi,
\end{equation}
where commutation of all logical operators with $\Pi$ was used. Since $P^2=Q^2=1$ as they are Pauli operators, we get $c=\pm1$. Thus, Eq.~\eqref{K_is_logical} becomes
\begin{equation}
P(UQU^\dag)\Pi=\pm(UQU^\dag)P\Pi,
\end{equation}
representing commutation or anticommutation of $P$ and $UQU^\dag$ with respect to the codespace. This holds for all $P$ and $Q$, implying $UQU^\dag$ is a logical Pauli operator for all $Q$, and thus $U$ is a logical Clifford operator by definition.

It remains to show that Eq.~\eqref{K_is_logical} holds, and it is here we use the bound on the range of gates in $U$. Let $\mathcal{B}(A,r)$ denote the set of qubits within distance $r$ of a set of qubits $A$. If $U$ has depth $h$ and consists of gates with range $R_x$ and $R_y$ in the $x$- and $y$-dimensions, then $V:=UQU^\dag\in\mathcal{B}(\xi_2,hR_y)\cap\mathcal{B}(\zeta_2,hR_x):=\beta_2$. Gates of $U$ outside the ``lightcone'' of $\alpha_2$ act trivially on $Q$ and cancel, leading to $V$ supported in a limited area. Now $V$ is a depth $(2h+1)$ circuit of bounded range gates, and so $K=PVPV^\dag$ is similarly supported only in the region $\mathcal{B}(\xi_1,(2h+1)R_y)\cap\mathcal{B}(\zeta_1,(2h+1)R_x):=\beta_1$. Yet, at the same time, because $V\in\beta_2$, we have $K\in\beta_2$. Thus, $K\in\beta_1\cap\beta_2$, a region which has size upper bounded by 
\begin{equation}
|\beta_1\cap\beta_2|<O(h^2R_xR_y+hR_x+hR_y).
\end{equation}
If $h^2R_xR_y+hR_x+hR_y<O(d)$, then $K$ can only be the trivial logical operator and so Eq.~\eqref{K_is_logical} is proved. In the case that $U$ has constant depth $h=O(1)$, then $(R_x+1)(R_y+1)<O(d)$ is sufficient for $K$ to be trivial.
\end{proof}

We note briefly that 3D Bacon-Shor codes \cite{Bacon2006} in the $Z$-gauge can perform $\overline{\text{CCZ}}$ without ranged gates, as long as the three codeblocks are allowed to be oriented differently. Symmetric 3D Bacon-Shor codes $m\times m\times m$ are constructed from $m$ $m\times m$ planes of the 2D Ising model with nearest neighbor $ZZ$ interactions ($XX$ interactions exist between adjacent planes). Let $\hat n_j$ be the vector perpendicular to all these $m$ planes for codeblock $j$. Taking $\hat n_1=\hat x$, $\hat n_2=\hat y$, and $\hat n_3=\hat z$ is sufficient for the layout to support $\overline{\text{CCZ}}$ without ranged gates. The generalization to local $\overline{\text{C}^k\text{Z}}$ gates in $k+1$-dimensional Bacon-Shor codes is straightforward (including $k=1$).


\section{The conditions for Pauli recovery and fault-tolerant circuits that violate them}\label{pauli_cond}
In this section, our goal is to decide when circuits of CCZ gates between Bacon-Shor codeblocks are fault-tolerant, and, in particular, if they are able to be made fault-tolerant by using only the ``standard'' stabilizer code recovery: a projective measurement of all stabilizers followed by a classically controlled recovery chosen from the Pauli group $\mathcal{P}$. Actually, we define two special-case recovery procedures -- stabilizer projective recovery (SPR) and stabilizer projective Pauli recovery (SPPR) -- the latter (a subset of the former) corresponding to the stabilizer code ``standard''. We then elucidate when SPR and SPPR exist for a given error channel on a stabilizer code. Next, we define CCZ-form circuits and prove a simpler characterization of the existence of SPR for them. Then, we bound the asymmetry of Bacon-Shor codes required such that CCZ-form circuits implementing $\overline{\text{CCZ}}$ are fault-tolerant using only SPR. Finally, we discuss the fault-tolerance (using non-Pauli recovery) of $\overline{\text{CCZ}}$ designs on Bacon-Shor codes.
Interestingly, while our recovery is not SPR, it still borrows most of its circuitry from that class of recovery.

We begin by formally defining SPR and SPPR.
\begin{defn}
For a stabilizer code with generators $\{s_k\}_{k\in[r)}$ define the set of orthogonal projectors indexed by $\alpha\in\{0,1\}^r$ as
\begin{equation}
P_{\alpha}=\prod_{k=0}^{r-1}\frac{I+(-1)^{\alpha_k}s_k}{2}.
\end{equation}
Then a stabilizer projective recovery (SPR) $\mathcal{R}$ is defined
\begin{equation}
\mathcal{R}(\sigma)=\sum_{\alpha\in\{0,1\}^r}\mathcal{R}_{\alpha}\left(P_\alpha\sigma P_\alpha\right).
\end{equation}
for arbitrary quantum channels $\mathcal{R}_\alpha$. If $\mathcal{R}_\alpha(\rho)=U_\alpha\rho U_\alpha^\dag$ with $U_\alpha\in\mathcal{P}$ for all $\rho$ and all $\alpha$, then the recovery is stabilizer projective with Pauli recovery (SPPR).
\end{defn}
That is, both SPR and SPPR assume that a complete set of stabilizers is measured, followed by a classically controlled channel. In SPPR this channel is simply a Pauli operator. Stabilizer codes and logical gates on stabilizer codes traditionally use SPPR, as this is sufficient for fault-tolerance of transversal operations.

However, SPR and SPPR is inherently weaker than the entire class of recovery allowed by the Knill-Laflamme conditions.
\begin{thm}[Knill-Laflamme \cite{Knill1997}]\label{KL}
Consider a quantum code with projector $P$ and quantum operation $\mathcal{E}$ with Kraus operators $\{E_j\}$. There exists a recovery operation $\mathcal{R}$ correcting $\mathcal{E}$ (i.e.~$\mathcal{R}(\mathcal{E}(\rho))\propto\rho$ for all $\rho$ in the codespace) if and only if
\begin{equation}\label{PEEP}
PE_i^\dag E_jP=\gamma_{ij}P,
\end{equation}
for a Hermitian matrix $\gamma$.
\end{thm}
For the proof of this famous theorem, we refer to \cite{Knill1997,Nielsen2000}. We note here a subtlety, however. In general, a trace-preserving error channel $\mathcal{E}$ is never exactly correctable, because it can conceivably involve catastrophic yet very low probability events. Indeed, it is more reasonable that Eq.~\eqref{PEEP} is merely approximately satisfied. To make this concrete, introduce an error parameter $\epsilon$ (e.g.~this could be the depolarizing error rate), and demand that
\begin{equation}\label{approxPEEP}
PE_i^\dag(\epsilon) E_j(\epsilon)P=\gamma_{ij}(\epsilon)P+O(\epsilon^d).
\end{equation}
We say the code has effective distance $d$ (with respect to the error channel $\mathcal{E}$) if this approximate condition holds. In \cite{Leung1997} it is argued that Eq.~\eqref{approxPEEP} is sufficient for the existence of a recovery operation $\mathcal{R}$ such that the fidelity $\mathcal{F}$ is bounded like
\begin{equation}
\mathcal{F}:=\min_{\bar\psi\in\mathcal{C}}\bra{\bar\psi}\left(\mathcal{R}\circ\mathcal{E}(\ket{\bar\psi}\bra{\bar\psi})\right)\ket{\bar\psi}\ge1-O(\epsilon^d).
\end{equation}
The minimization is over all states in the codespace $\mathcal{C}$. 

As an example, single-qubit depolarizing noise over all $n$ qubits in a code has Kraus operators 
\begin{equation}
\left\{\epsilon^{|\vec v|}\sigma^{v_1}_1\otimes\sigma^{v_2}_2\dots\otimes\sigma^{v_n}_n:\vec v\in\{0,1,2,3\}^n\right\},
\end{equation}
where subscripts indicate the affected qubit and superscripts the Pauli operator. The Hamming weight $|\vec v|$ is defined as the number of non-zero elements of $\vec v$. In this case, effective distance $d$ implies error-correction of all errors of weight $<d$, as per the definition of code distance. Later, we consider the effective distance of circuit noise channels, those in which each Kraus operator is a collection of circuit faults that have propagated to the end of the circuit.

Generalizing similar ideas to entanglement fidelity has yielded both necessary \emph{and} sufficient conditions for approximate correction \cite{Beny2010}. While these ideas are likely applicable here, for simplicity we stick with the notion of approximate correction in Eq.~\eqref{approxPEEP}.

Our immediate goal, however, is to develop conditions analogous to Knill-Laflamme for performing SPR.
\begin{thm}\label{SPR}
Given a stabilizer code with projector $P$ and a quantum operation $\mathcal{E}$ with Kraus operators $\{E_j\}$, there exists a SPR $\mathcal{R}$ correcting $\mathcal{E}$ if and only if there exist hermitian matrices $\gamma^{(\alpha)}$ such that for all $\alpha,i,j$
\begin{align}\label{PEPEP}
PE_i^\dag P_\alpha E_jP=\gamma^{(\alpha)}_{ij}P.
\end{align}
\end{thm}
\begin{proof}
Our arguments closely follow \cite{Nielsen2000}. Indeed, we might consider $P_\alpha E_j$ to be the operator elements of a new error channel and apply Theorem~\ref{KL}. Yet, this does not necessarily guarantee the recovery is SPR as we have defined it, and so we run through the complete proof for definitiveness. Note that because $\sum_\alpha P_\alpha=I$, summing Eq.~\eqref{PEPEP} over $\alpha$ implies Eq.~\eqref{PEEP}; SPR is indeed a subset of Knill-Laflamme recovery.

We prove the reverse direction first. Since each $\gamma^{(\alpha)}$ is hermitian, each can be diagonalized $d^{(\alpha)}=u^{(\alpha)\dag}\gamma^{(\alpha)}u^{(\alpha)}$. Define unitary equivalent Kraus operators $F^{(\alpha)}_k=\sum_ju^{(\alpha)}_{jk}E_j$. That is, $\mathcal{E}^{(\alpha)}$ with Kraus operators $\{F^{(\alpha)}_k\}$ satisfies $\mathcal{E}^{(\alpha)}(\rho)=\mathcal{E}(\rho)$ for all $\rho$ and $\alpha$. Notice furthermore that for all $\alpha$,
\begin{equation}
PF_k^{(\alpha)}P_\alpha F_l^{(\alpha)}P=d_{kl}^{(\alpha)}P,
\end{equation}
by using the definition of $F^{(\alpha)}_k$ and Eq.~\eqref{PEPEP}.

Use the polar decomposition $A=U\sqrt{A^\dag A}$ on $P_\alpha F_l^{(\alpha)}P$. This guarantees the existence of unitaries $U_k^{(\alpha)}$ such that
\begin{equation}
P_\alpha F^{(\alpha)}_kP=\sqrt{d^{(\alpha)}_{kk}}U_k^{(\alpha)}P.
\end{equation}
For any fixed $\alpha$, the projectors $P^{(\alpha)}_k=U^{(\alpha)}_kPU^{(\alpha)\dag}_k$ are orthogonal,
\begin{align}
P^{(\alpha)}_lP^{(\alpha)}_k&=\frac{U^{(\alpha)}_lPF^{(\alpha)\dag}_lP_\alpha F_kPU^{(\alpha)\dag}_k}{\sqrt{d^{(\alpha)}_{ll}d^{(\alpha)}_{kk}}}\\
&=\frac{d_{lk}^{(\alpha)}}{\sqrt{d^{(\alpha)}_{ll}d^{(\alpha)}_{kk}}}U^{(\alpha)}_lPU^{(\alpha)\dag}_k,
\end{align}
which is zero when $l\neq k$. If $\sum_kP_k^{(\alpha)}<I$ we can add another projector to complete the set and define
\begin{equation}
\mathcal{R}_\alpha(\rho)=\sum_kU^{(\alpha)\dag}_kP^{(\alpha)}_k\rho P^{(\alpha)}_kU^{(\alpha)}_k.
\end{equation}
We can now show $\mathcal{R}\left(\mathcal{E}(\rho)\right)\propto\rho$ for any $\rho=P\rho P$ in the codespace.
\begin{align}
\mathcal{R}\left(\mathcal{E}(\rho)\right)&=\sum_{\alpha}\mathcal{R}_\alpha\left(P_\alpha\mathcal{E}(\rho) P_\alpha\right)\\\nonumber
&=\sum_{\alpha}\mathcal{R}_\alpha\left(P_\alpha\mathcal{E}_\alpha(\rho) P_\alpha\right)\\\nonumber
&=\sum_\alpha\sum_{kl}PU^{(\alpha)\dag}_kP_\alpha F^{(\alpha)}_lP\rho PF^{(\alpha)\dag}_l P_\alpha U^{(\alpha)}_kP\\\nonumber
&=\sum_\alpha\sum_{kl}d^{(\alpha)}_{kl}\rho\propto\rho.
\end{align}

For the forward direction, we notice that $\mathcal{R}(\mathcal{E}(P\rho P))$ defines a channel for all $\rho$ (not just $\rho$ in the codespace). By the assumption that $\mathcal{R}$ corrects $\mathcal{E}$, we have
\begin{equation}\label{assumption}
\mathcal{R}(\mathcal{E}(P\rho P))=c P\rho P.
\end{equation}
Linearity guarantees $c$ does not depend on $\rho$. Now, Eq.~\eqref{assumption} holds for all $\rho$ and therefore the Kraus operators of the channel on the left and the channel on the right must be unitary related. This means there are constants $c^{(\alpha)}_{lk}$ such that
\begin{equation}
R_l^{(\alpha)}P_\alpha E_kP=c^{(\alpha)}_{lk}P
\end{equation}
if $\{R^{(\alpha)}_l\}$ are the Kraus operators of $\mathcal{R}_\alpha$. Thus, using the completeness of these Kraus operators,
\begin{align}
\left(\sum_lc_{lj}^{(\alpha)*}c^{(\alpha)}_{lk}\right)P&=PE_j^\dag P_\alpha \left(\sum_lR^{(\alpha)\dag}_lR^{(\alpha)}_l\right)P_\alpha E_kP\\
&=PE_j^\dag P_\alpha E_kP.
\end{align}
Because the parenthesized term on the left is a hermitian matrix, this is what we set out to show.
\end{proof}

Theorem~\ref{SPR} says that SPR works if projection to the codespaces $P_\alpha$ does not destroy the orthogonality of the error operators $E_j$. Analogous to the notion of effective distance defined in Eq.~\eqref{approxPEEP}, we have a notion of effective distance $d$ using SPR when
\begin{equation}\label{approxPEPEP}
PE_i^\dag(\epsilon) P_\alpha E_j(\epsilon)P=\gamma^{(\alpha)}_{ij}(\epsilon)P+O(\epsilon^d).
\end{equation}
An SPR $\mathcal{R}$ exists such that the fidelity is bounded as $\mathcal{F}\ge1-O(\epsilon^d)$ when Eq.~\eqref{approxPEPEP} holds.

We also develop conditions for performing SPPR, the stabilizer code standard.
\begin{thm}\label{SPPR}
Given a stabilizer code with projector $P$ and a quantum operation $\mathcal{E}$ with Kraus operators $\{E_j\}$, there exists an SPPR $\mathcal{R}$ correcting $\mathcal{E}$ if and only if there exist constants $c_{\alpha k}$ and unitaries $U_\alpha\in\mathcal{P}$ such that for all $j,\alpha$,
\begin{align}\label{PFP}
P_\alpha E_jP=c_{\alpha j}U_\alpha P.
\end{align}
\end{thm}
\begin{proof}
We prove the reverse direction first. Using the set of Kraus operators $\{E_j\}$ and for $\rho=P\rho P$ in the codespace,
\begin{align}
\mathcal{R}(\mathcal{E}(\rho))&=\sum_\alpha\sum_jU_\alpha^\dag P_\alpha E_j\rho E_j^\dag P_{\alpha} U_\alpha,\\
&=\sum_\alpha\sum_jU_\alpha^\dag (P_\alpha E_jP)\rho (PE_j^\dag P_{\alpha}) U_\alpha,\\
&=\sum_\alpha\sum_k|c_{\alpha j}|^2U_\alpha^\dag U_\alpha\rho U_\alpha^\dag U_\alpha\\
&\propto \rho.
\end{align}

Now for the forward direction. We assume
\begin{equation}
\mathcal{R}(\mathcal{E}(P\rho P))=\sum_\alpha\sum_j U_\alpha^\dag P_\alpha E_jP\rho PE_j^\dag P_\alpha U_\alpha=cP\rho P,
\end{equation}
for constant $c$ with $U_\alpha\in\mathcal{P}$ for all $\alpha$. Because this holds for all $\rho$, the channel with one operator element $\sqrt{c}P$ must be unitarily equivalent to the one with elements $\{U^\dag_\alpha P_\alpha E_j P\}_{\alpha,j}$. This implies the existence of constants $c_{\alpha j}$ such that
\begin{equation}\label{unitary_equiv}
U^\dag_\alpha P_\alpha E_jP=c_{\alpha j}P.
\end{equation}
Multiplying both sides by $U_\alpha$, we get Eq.~\eqref{PFP}. It is worth noting that using any other unitarily equivalent set of operator elements for $\mathcal{E}$ will just result in linear combinations of Eq.~\eqref{PFP} over the index $j$ (and not $\alpha$), and so does not change the conclusions. Thus, to confirm SPPR correctability of a channel we only have to verify Eq.~\eqref{PFP} for one set of operator elements.
\end{proof}
Theorem~\ref{SPPR} says that a Pauli recovery operation can be used whenever projecting the operator elements of the error channel to the orthogonal codespaces gives a Pauli error depending only on the projection result. 

We now consider Theorems~\ref{SPR} and \ref{SPPR} in the context of CCZ-form circuits.
\begin{defn}
A CCZ-form circuit is composed entirely of CCZ gates, and moreover, the qubits can be partitioned into three sets (say $A_i$ for $i=1,2,3$) such that any one CCZ gate acts on at most one qubit from each set. Without loss of generality, assume no CCZ gate is repeated (otherwise, they could be canceled). A $\text{C}^k\text{Z}$-form circuit is defined analogously for $i=1,2,\dots,k+1$.
\end{defn}
To characterize when SPR is appropriate for CCZ-form circuits, we need to discuss how errors propagate through them. In general, circuits define ``lightcones'', which contain all qubits that may be correlated. Lightcones also bound the region that errors may propagate. To be concrete, let $C=\{S_1,S_2,\dots,S_h\}$ be a circuit broken into timesteps $S_j$, each timestep a set of gates with disjoint support. For a set of qubits $Q$ at time $t$ the forward lightcone of $Q$ is denoted $\mathcal{L}_{t}(Q)$. We can define this recursively,
\begin{align}
l_{v}(Q)&=\{i:\exists g\in S_v,j\in Q\text{ s.t. }\{i,j\}\subseteq\text{supp}(g)\},\\
\mathcal{L}_{t}(Q)&=\mathcal{L}_{t+1}(l_{t+1}(Q)),
\end{align}
with $\mathcal{L}_h(Q)=Q$. The lightcone of a gate $g\in S_t$ is defined as the lightcone of that gate's output qubits, $\mathcal{L}(g)=\mathcal{L}_t\left(\text{supp}(g)\right)$.

However, CCZ-form circuits already restrict the propagation of errors more severely than na\"{i}ve application of lightcones would suggest. Indeed, it is not hard to verify the following claim.
\begin{claim}\label{claim}
In a CCZ-form circuit with depth $h$, the failure of a CCZ gate $g\in S_t$ places
\begin{enumerate}
\item at most one $X$ error per $A_i$
\item at most $Z$ errors on all qubits in the \emph{modified lightcone}
\begin{equation}
\tilde{\mathcal{L}}\left(g\right):=\bigcup_{v=t+1}^hl_v\left(\text{supp}(g)\right).
\end{equation}
\end{enumerate}
\end{claim}
\noindent This follows from the fact that a CCZ gate commutes with Pauli $Z$, but upon an incoming $X$ error on one node propagates CZ between its other two nodes.

We now specialize to CCZ-form circuits on three 2D Bacon-Shor codeblocks. The codeblocks define the qubit partitioning $\{A_i\}$. To argue for fault-tolerance, we assume circuit depolarizing noise, i.e.~a gate fails (called a \emph{fault}) with probability $p$ by applying any Pauli error on its support following the ideal application of the gate. An error-channel $\mathcal{E}$ can be defined as acting on the output qubits with Kraus operators $\{E_j\}$ representing the products of depolarizing-noise generated Pauli errors propagated to the end of the circuit. A Kraus operator representing $t$ faults has order $O(p^{t/2})$. The error parameter $\epsilon=\sqrt{p}$ is used to define effective distance via Eq.~\eqref{approxPEPEP} above. In words, effective distance $d$ implies that up to $d-1$ faults are detectable, and any $\lfloor(d-1)/2\rfloor$ faults are correctable.
\begin{lem}\label{bs_ccz_lem}
Let $m,n$ be integers with $n\ge m$ and $m\ge3$. Depolarizing noise in a CCZ-form circuit on three $m\times n$ $Z$-gauge Bacon-Shor codeblocks is correctable with effective distance $m$ using SPR if and only if all gates have modified lightcones (see Claim~\ref{claim}) that intersect any codeblock in at most two rows.
\end{lem}
\begin{proof}
For the forward direction, we assume by way of contradiction that some gate has a lightcone intersecting codeblock $A_i$ on at least three qubits. We appeal to Theorem~\ref{SPR} and find two $E_j$ and a codespace projector $P_\alpha$ that fail Eq.~\eqref{approxPEPEP}, implying that the circuit is not effective distance $m$, thereby obtaining a contradiction.

Label rows in $A_i$ by $r_1,r_2,\dots,r_m$ and without loss of generality assume the first three intersect the lightcone of gate $g$. Let $E_1$ represent the set of faults containing $g$ failing with $XXX$ on its support and single-qubit $Z$ errors on each of $r_4,r_5,\dots,r_{\lceil m/2\rceil+1}$. This is a total of $\lceil m/2\rceil-1$ faults. Let $E_2$ represent the set of faults again containing $XXX$ after $g$, but also single-qubit $Z$ errors on $r_{\lceil m/2\rceil+2},\dots,r_{m}$. This is a total of $\lfloor m/2\rfloor$ faults. Note $E_1E_2$ is the result of $m-1<m$ faults, and thus has order $O(p^{(m-1)/2})$.

Choose $P_\alpha$ corresponding to violation (i.e.~projection onto the $-1$-eigenspace) of the $Z$-type stabilizers that indicate $XXX$ on $\text{supp}(g)$ and violation of the $X$-type stabilizer spanning rows $r_{\lceil m/2\rceil+1}$ and $r_{\lceil m/2\rceil+2}$. All other stabilizers of any codeblock are not violated. 

With this setup, we see $PE_1^\dag P_\alpha\propto PE_1'P_\alpha=PE'_1$, where
\begin{equation}
E_1'\propto\prod_{q\in\text{supp}(g)}X_q\prod_{j=1}^{\lceil m/2\rceil+1}Z_{r_j}
\end{equation}
in which $X_q$ is an $X$ on qubit $q$ and $Z_{r_j}$ indicates $Z$ on any qubit in row $r_j$ (they are all equivalent with respect to $P_\alpha$ or $P$). By way of explanation, while propagating $XXX$ on the support of $g$ introduces CZ errors, we can collapse these CZ errors to Paulis using the projector $P_\alpha$. Our choice of $P_\alpha$ ensures $Z_{r_1}Z_{r_2}Z_{r_3}$ is the sole result of this collapse. 

Likewise, $P_\alpha E_2P\propto P_\alpha E'_2P=E'_2P$ with
\begin{equation}
E_2'\propto\prod_{q\in\text{supp}(g)}X_q\prod_{j=\lceil m/2\rceil+2}^{m}Z_{r_j}.
\end{equation}
Now notice $PE_1^\dag P_\alpha E_2P\propto PE'_1E'_2P$. However, this is not proportional to $P$ because $E'_1$ and $E'_2$ differ by a logical operator $\bar Z=\prod_{j=1}^mZ_{r_j}$. So Eq.~\eqref{approxPEPEP} fails to hold.

For the reverse direction, it is enough to satisfy Eq.~\eqref{approxPEPEP} to notice that a logical error cannot be written onto the data with any combination of fewer than $m$ faults. We do this using Claim~\ref{claim}. First, $\bar X$ cannot be created because the minimum weight of $\bar X$ is $n\ge m$ and each CCZ gate failure introduces at most one $X$ error per block. It remains to argue that $Z$ errors cannot cause $\bar Z$.

We do this by arguing that it always takes $t$ faults for every $t$ $Z$ errors placed in a specific codeblock, say $A_i$. Recall our goal is to cause $\bar Z_i$ exactly with less than $m$ faults. So every $X$ error we introduce (and we must introduce at least one, otherwise it clearly takes $m$ $Z$ errors to cause $\bar Z_i$) must also be removed by an \emph{additional} fault later in the circuit. While a single faulty CCZ gate, failing with some correlated $X$ errors on its support, may introduce $\le2$ $Z$ errors to $A_i$ (because of its restricted modified lightcone, see Claim~\ref{claim}) those $X$ errors must then be erased via a later fault. Two faults lead to at most two $Z$ errors. There is one other case to worry about though, when three faults can lead to three $Z$ errors. If $E_1$ and $E_2$ are the $X$ errors introduced by two different CCZ gates $g_1$ and $g_2$ failing, the ability to erase $E_1E_2$ by failure of a single later CCZ $g'$ implies that the union of the modified lightcones $\tilde{\mathcal{L}}(g_1)\cup\tilde{\mathcal{L}}(g_2)$ intersects $A_i$ on at most $3$ qubits, not $4$. This is so because $\tilde{\mathcal{L}}(g')\cap A_i$ (which is not empty) is a subset of both $\tilde{\mathcal{L}}(g_1)\cap A_i$ and $\tilde{\mathcal{L}}(g_2)\cap A_i$.
\end{proof}

We also want some guarantee on the size of circuits for logical operators on Bacon-Shor codes. This might be thought of a more detailed version of Appendix~\ref{BK-BS}.
\begin{lem}\label{circ_size}
A $\text{C}^k\text{Z}$-form circuit implementing $\overline{\text{C}^k\text{Z}}$ on $m\times n$ 2D Bacon-Shor codes in the $Z$-gauge must use $m^{k+1}$ gates. Moreover, selecting a row from each codeblock, there is exactly one gate joining qubits from all those rows.
\end{lem}
\begin{proof}
The proof proceeds inductively. Begin with $\overline{\text{CZ}}$. Assume each row is involved in $\le m$ CZ gates, otherwise they could be canceled via $Z$-gauge operators. If fewer than $m^2$ gates were used in a CZ-form circuit $C_1$, then there are rows $r_{A_1}$ and $r_{A_2}$ from codeblocks $A_1$ and $A_2$ that are not connected by a CZ gate. Thus, when $\bar X$ supported on $r_{A_1}$ is conjugated by $C_1$ we get $\bar X$ times at most $m-1$ Pauli $Z$s on codeblock $A_2$. This is not enough to construct $\bar Z$ on codeblock $A_2$, so the circuit $C_1$ cannot implement $\overline{\text{CZ}}$.

For the induction, assume $\overline{\text{C}^{k-1}\text{Z}}$ uses $m^k$ $\text{C}^{k-1}\text{Z}$ gates. If $\overline{\text{C}^k\text{Z}}$ could be implemented with fewer than $m^{k+1}$ gates by circuit $C_k$, a gate would be missing that couples some set of rows $\{r_{A_1},r_{A_2},\dots,r_{A_{k+1}}\}$, one from each codeblock. Then, when $\bar X$ supported on $r_{A_1}$ is conjugated by $C_k$, we get $\bar X$ times a $\text{C}^{k-1}\text{Z}$-form circuit with at most $m^k-1$ gates, which is not enough for $\overline{\text{C}^{k-1}\text{Z}}$ by the inductive assumption. So, $C_k$ cannot implement $\overline{\text{C}^k\text{Z}}$.
\end{proof}

With Lemmas~\ref{bs_ccz_lem} and \ref{circ_size} in hand, we can argue that CCZ-form circuits implementing $\overline{\text{CCZ}}$ on Bacon-Shor codes cannot be fault-tolerant using SPR if the code is too symmetric.
\begin{thm}
A CCZ-form circuit for $\overline{\text{CCZ}}$ on $m\times n$ 2D Bacon-Shor codes that is fault-tolerant to circuit depolarizing noise with SPR exists if $n\ge\lceil m/2\rceil^2$ and only if $n\ge\lceil m^2/4\rceil$.
\end{thm}
\begin{proof}
The ``if" claim is constructive and uses the 2-transversality idea of \cite{Yoder2016}. For each codeblock $A_i$, partition the rows into $\lceil m/2\rceil$ sets $S_j^{(i)}$ of size at most two (i.e.~$\lfloor m/2\rfloor$ pairs and one unpaired row if $m$ is odd). There are $\lceil m/2\rceil^3$ ways to choose a tuple $(S_j^{(1)},S_k^{(2)},S_l^{(3)})$, and with each we will associate a CCZ-form circuit $C_{jkl}$ with $\le8$ CCZ gates that connects all rows in the sets in all ways (i.e.~is ``round-robin" as in \cite{Yoder2016}). For each column number $c\in[n)$, choose $\lceil m/2\rceil$ of $C_{jkl}$ that are disjoint. Since there are $n\ge\lceil m/2\rceil^2$ columns, we can do all the $C_{jkl}$ in parallel. Importantly, the construction of $C_{jkl}$ guarantees no gate has modified lightcone intersecting any codeblock in more than two rows. So, Lemma~\ref{bs_ccz_lem} guarantees fault-tolerance with SPR. In fact, it is not hard to show SPPR is sufficient as well. Once $t$ $X$ errors are located (from the $Z$-gauge syndrome data) the possible locations of at most $2t$ $Z$ errors per block are also located.

For the ``only if'' claim, we notice that any one qubit involved in $>4$ CCZ gates is pigeonholed to have a modified lightcone including at least three rows of a codeblock. Moreover, each row is involved in $m^2$ gates, which means there exists a qubit involved in at least $\lceil m^2/n\rceil$ gates. Thus, appealing to Lemma~\ref{bs_ccz_lem}, the bound $\lceil m^2/n\rceil\le4$ is necessary for SPR to exist. Note that for even $m$, this bound is equivalent to $n\ge\lceil m/2\rceil^2$.
\end{proof}

It remains to describe the fault-tolerance of our $3\times 3$ code using non-SPR recovery. Recall the circuit Eq.~\eqref{CCZ33}. In Fig.~\ref{cczpiece} we show the timesteps for column $j=0$. Since CCZ gates are isolated to single columns and other columns are related by qubit permutations within the column we need only consider fault-tolerance of the construction for column $j=0$.

Note first that Lemma~\ref{bs_ccz_lem} is violated. Any gate in timestep $t=0$ has a modified lightcone spanning all rows of codeblocks A and C. These are the only gates violating the Lemma however, so we need only check that all errors (in particular the $X$ errors) introduced at $t=0$ are correctable. The circuit is already wired in such a way that \emph{single} $X$ errors cannot propagate $Z$s to more than three rows of any block. In the case of correlated $X$ errors, the decoder (after measuring just the $Z$-gauge) knows a CCZ failed. If there is one $X$ error per block, it knows exactly which failed and can apply suitable $X$ and CZ correction. If two $X$ errors are present, the decoder knows one of two CCZ gates failed. However, we have built the circuit so that the pair of $X$ errors has a suitably restricted modified lightcone ($X$ errors on blocks A and C) or such that CZs can be unambiguously applied to correct part of the error propagation (for instance, it is sufficient to correct the CZ errors resulting from $t=2$ if the possible faulty CCZs are known to be in either $t=0$ or $t=1$). Remaining $Z$ errors are now located to at most two rows of each codeblock, and measurement of the $X$-stabilizers is sufficient to correct them. See also Appendix~\ref{thr_detail} for an explicit description of the decoder.

\begin{figure}
\includegraphics[width=0.8\columnwidth]{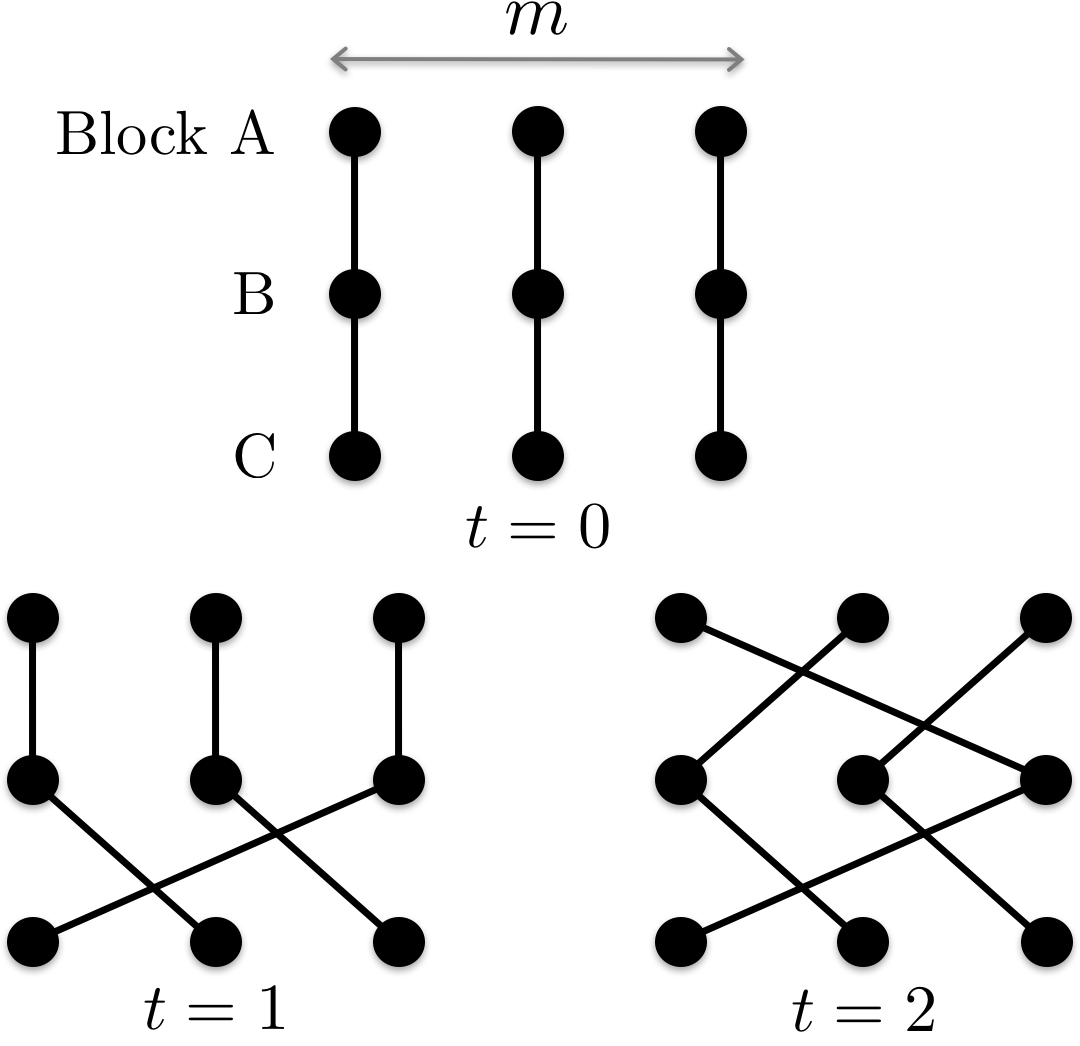}
\caption{\label{cczpiece}A diagram of three timesteps ($t=0,1,2$) of CCZ gates from the circuit in Eq.~\eqref{CCZ33}, shown for column $j=0$.}
\end{figure}

\section{Circuits, volume, and logical gate times in MUSICQ}\label{gate_detail}
In this section, we present circuits for the magic state injection protocols considered in Table~\ref{tab1}. We also discuss the circuit volume as a metric for comparing logical gates. Finally, we use projected physical gate times from the ion trap MUSICQ architecture \cite{Monroe2014} to estimate the spacetime volume (in units of $\mu s\times$qubits) and total time of logical CCZ gates.

The magic state used to implement CCZ is
\begin{align}
|\text{CCZ}\rangle&=\frac{1}{2\sqrt2}\sum_{i,j,k\in\{0,1\}}(-1)^{ijk}|ijk\rangle\\
&=\text{CCZ}\ket{+}^{\otimes3}.
\end{align}
This is the $+1$-eigenstate of the three stabilizers
\begin{equation}\label{CCZ_stabs}
S_1=X_1\text{CZ}_{23},\quad S_2=X_2\text{CZ}_{13},\quad S_3=X_3\text{CZ}_{12}.
\end{equation}
Notice $\ket{+}\ket{+}\ket{0}$ is already stabilized by $S_1$ and $S_2$, so we need only measure $S_3$ to create $\ket{\text{CCZ}}$.

This is what is done logically in Fig.~\ref{createCCZ}. This circuit prepares $\ket{\overline{\text{CCZ}}}$ fault-tolerantly for any distance three code with transversal CZ. We also need a way to prepare CAT states tolerant to one fault. This is done by Fig.~\ref{createCAT}. Using the $\ket{\overline{\text{CCZ}}}$ state to implement CCZ on three codeblocks is done using Fig.~\ref{injectCCZ}.


\begin{figure}[t]
\includegraphics[width=\columnwidth]{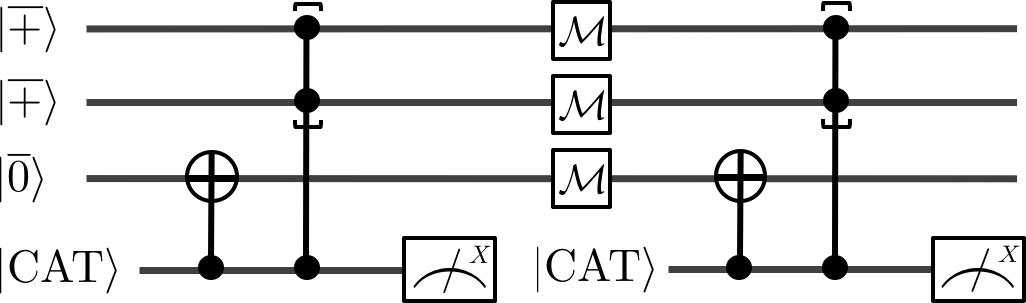}
\caption{\label{createCCZ} Preparing $\ket{\overline{\text{CCZ}}}$. The bracketed CZ indicates that logical transversal CZ is to be transversally controlled on the CAT state. CAT states are as large as a codeblock. The measurements are transversal, but the only meaningful result is the parity. Between measurements we perform full error-correction $\mathcal{M}$ on the codeblocks. If both parity measurements are the same, we accept, and if they are both 1, we need to apply $\overline{Z}$ to the last codeblock.}
\end{figure}

\begin{figure}
\includegraphics[width=0.8\columnwidth]{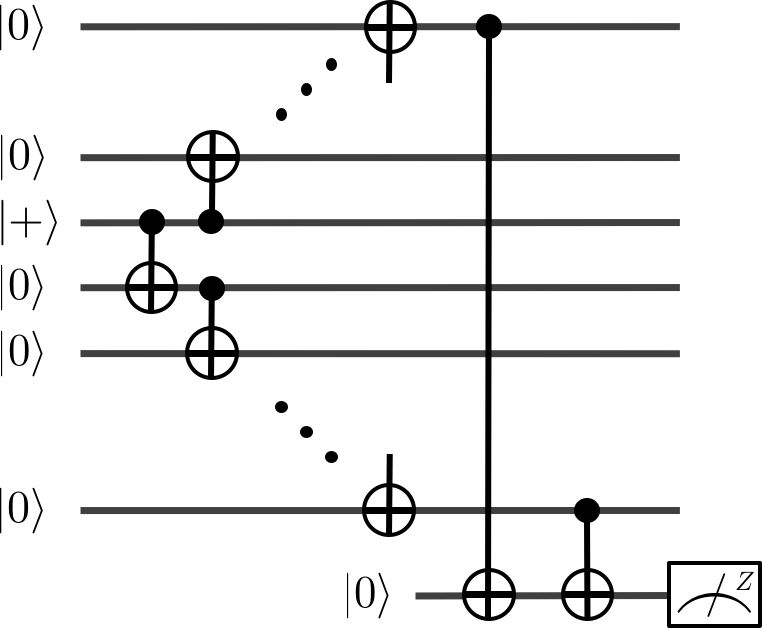}
\caption{\label{createCAT} Preparing a CAT state that will have at most a single-qubit error if we postselect on the measurement returning 0.}
\end{figure}

\begin{figure}[h]
\includegraphics[width=\columnwidth]{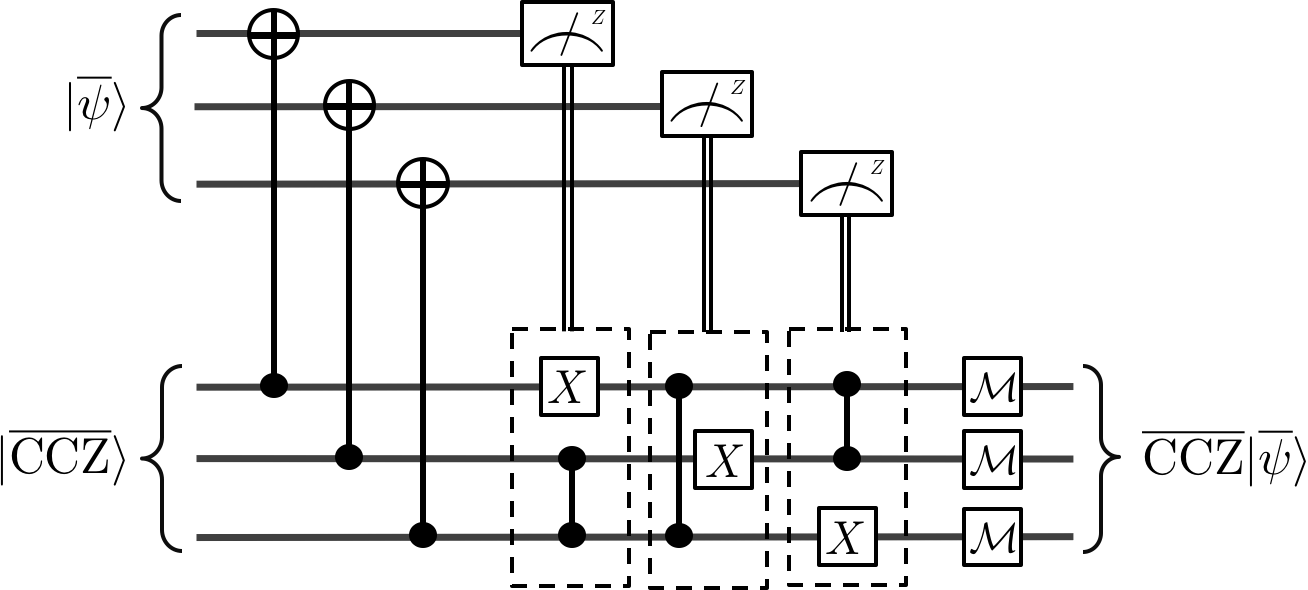}
\caption{\label{injectCCZ} A $\ket{\overline{\text{CCZ}}}$ state is used to implement $\overline{\text{CCZ}}$ fault-tolerantly. The logical measurements control logical implementations of the stabilizers $S_i$ from Eq.~\eqref{CCZ_stabs}. The final step $\mathcal{M}$ is complete stabilizer error-correction on all codeblocks.}
\end{figure}

However, relying on these magic-state constructions for implementing logical gates on low-distance codes is not necessarily a good idea. We now establish some metrics for comparing these circuits. Evaluating the magic-state constructions and our $3\times3$ Bacon-Shor CCZ, we find the latter decidedly advantaged; see Table~\ref{tab1}.

While qubit count and circuit depth are two simple circuit metrics, they are more easily manipulated. For example, it is well-known that a single CAT state can be repeatedly prepared, coupled to the data, and measured to extract all stabilizer syndromes. However, this takes a long time and thresholds suffer. A better metric should combine time and space. The circuit volume $\text{CV}$ does exactly this. If $n$ qubits are used and qubit $j$ is active for $s_j$ timesteps (i.e.~is involved in $s_j$ state initializations, gates, and measurements) then $\text{CV}=\sum_{j=1}^ns_j$.

We can fine-tune this metric if we know how long each circuit component takes on the physical hardware, e.g.~MUSICQ. Then, if a qubit $j$ is active for $t_j$ time, we define the spacetime volume $\text{ST}=\sum_{j=1}^nt_j$.

Both $\text{CV}$ and $\text{ST}$ are now relatively easily calculated for the magic-state circuits pictured here and our small $\overline{\text{CCZ}}$ on the $3\times3$ Bacon-Shor code in the MUSICQ architecture. MUSICQ single-qubit gates take time $1\mu s$ and two and three-qubit gates take $10\mu s$. State-preparation takes $1\mu s$ and measurement takes $30\mu s$ \cite{Monroe2014}. These are not state-of-the-art numbers but rather idealized times. We rounded the numbers in Table~\ref{tab1}, not because the counting is inexact but because small changes to the circuits and how they are parallelized will change the minor digits. For instance, we assume ancilla code blocks are available immediately when needed, rather than having to wait to reinitialize them (this is in contrast to the MUSICQ time calculation next). In any case, a rough counting suffices to distinguish our construction from the magic-state constructions.

In the interest of learning something about how fast fault-tolerant quantum computers might actually be, our last comparison regards the total time of logical gates, assuming fixed qubit counts. The number of qubits influences how parallelized circuits, such as Steane error-correction Fig.~\ref{GaugedSteaneEC}, can be. It is a somewhat arbitrary choice, but we attempt consistency by taking just one ancilla codeblock for each codeblock that needs error-correction (implying, for instance, the parts of Fig.~\ref{GaugedSteaneEC} are done in series with no parallelization). It also happens that the constructions then fit within a 100 qubit elementary logical unit (ELU) \cite{Monroe2014}. For our Bacon-Shor construction for example, this means $3\times 9$ data qubits and $3\times 9$ ancillas, a total of 54. For magic-states on the 9-qubit code we need $2\times3\times 9$ ancillas to hold the $\ket{\overline{\text{CCZ}}}$ state and do the error-correction in Fig.~\ref{createCCZ}, and we need $3\times9$ more for the data, a total of $81$ qubits. For magic-states applied to the 7-qubit code, we use Goto's method \cite{Goto2016} of preparing Steane states with one verification ancilla, so that $3\times 7+3\times(7+1)$ ancillas suffice and $3\times7$ data qubits, for a total of $66$ qubits.

For the magic-states we count both the time for creating $\ket{\overline{\text{CCZ}}}$ and the time for injecting it. Only twelve multi-qubit gates can be done in parallel in (our assumed implementation of) MUSICQ, which extends some timesteps (e.g.~when three codeblocks couple to Steane states at once).

\section{Pseudothreshold details}\label{thr_detail}
Here we clarify the details of our exREC \cite{Aliferis2006} pseudothreshold calculations for $3\times3$ Bacon-Shor, which are performed using exact counting. While the calculation of pseudothresholds for identity, Hadamard, and CNOT are essentially standard exREC calculations, the CCZ calculation must handle non-Clifford gates. We do this by tracking errors as Pauli sums, keeping phase coherence. Although not scalable, this method is sufficient for the relatively small $3\times3$ Bacon-Shor code. Also, we should note that our identity exREC calculations are performed identically to those in \cite{Yoder2017} (except with Steane error-correction), so are directly (and favorably) comparable with the identity exRECs considered there.

Our exRECs are formed from three components, a leading error-correction (LEC), the logical gate (Ga), and the trailing error-correction (TEC). The exREC is the composition of these in order TEC.Ga.LEC, such that LEC acts first. We also use the concept of an ideal decoder (Id), which measures all stabilizers and applies noiseless recovery, to define exREC failure. An exREC fails if the ideally decoded state following the TEC does not match the expected state given the ideally decoded state after the LEC. Thus, failure as defined requires at least one fault to be present in the Ga or TEC, and the LEC is included simply to model incoming errors.

\begin{figure}
\includegraphics[width=\columnwidth]{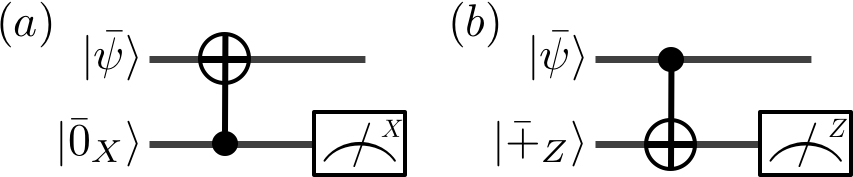}
\caption{\label{GaugedSteaneEC} (a) Measuring the $X$-type gauge operators with an $X$-gauge logical-zero state, transversal CNOT, and transversal measurement in the $X$-basis. (b) Likewise, measuring $Z$-type gauge operators. For a $Z$-gauge codeblock, type-1 error-correction is (b) followed by (a), while type-2 error-correction is (a) followed by (b).}
\end{figure}

In our calculations, Steane error-correction is used to extract syndromes via Fig.~\ref{GaugedSteaneEC}. We always perform type-1 correction, that is, ordering parts (a) and (b) of Fig.~\ref{GaugedSteaneEC} such that the Bacon-Shor codeblocks change gauge (from $X$ to $Z$ or vice-versa) after syndrome extraction. As described in the text, this gives us more information about the errors. Moreover, the LEC is always built to go from the $X$-gauges to the $Z$-gauge and the TEC from the $Z$ to the $X$. That way, all gates Ga take place on $Z$-gauge Bacon-Shor. While these choices are by symmetry irrelevant for the transversal Clifford gates, our CCZ construction works assuming the $Z$-gauge. To be consistent, the Id also always uses type-1 correction.

Once the syndrome is extracted, we need to decide on a recovery to perform, a process called decoding. This we perform by table lookup. All LEC, TEC, and Id (except for the CCZ TEC, which we describe separately) follow the same basic scheme. For each pattern of stabilizer measurements, determine the lowest weight Pauli error that is consistent. In particular, we make no attempt to optimize recovery over any structure of the circuit, an approach that may offer marginal, but not at all substantial improvement. Since we want to end in the +1-eigenspace of the new gauge (e.g.~$Z$), we should also apply gauge operators (e.g.~$\bar X_{i,j}$) based on the final set of gauge measurements (e.g.~of $\bar Z_{i,j}$) to ensure this.

The decoder for the CCZ TEC is more complicated. We first measure $Z$-gauge operators. Assuming at most one fault, we learn the locations of all $X$ errors, at most one per codeblock. The locations of $X$ errors define a set $C$ of CCZ gates, such that the failure of any one is capable of causing the $X$ errors. Find the temporally last CCZ from the set $C$ and assume it failed. Apply a recovery of $X$ and CZ gates to correct that failure. Also note that if one of the earlier CCZ gates failed instead, there could still be $Z$ and CZ errors remaining on the data. Record all rows of all codeblocks that could still be affected by these errors. Next, $X$-gauge operators are measured. The $X$-stabilizer information, along with the recorded possible locations of $Z$-errors, allows us to correct any remaining $Z$-errors (assuming one fault). If there are two or more faults, this procedure will necessarily fail on some cases. For us, if the set $C$ is empty, we default to the usual TEC used for logical transversal Clifford gates. While this decoding has a lengthy description in words, it can still be precomputed as a simple decoding table, and therefore the required classical computational overhead is just as little as for Clifford TEC.

Having described the exREC circuits and decoding, we now describe our simulation. Our error-model is standard circuit depolarizing noise -- a $q$-qubit gate $g$ (including identity $I$) is assigned a probability $p_g$ of failing, and when it does each of the $(4^q-1)$ non-identity, $q$-qubit Pauli errors has a probability $p_g/(4^q-1)$ of occurring. Initialization and measurement in Pauli bases are slightly different in that they are unaffected by one type of Pauli error (e.g. initialization of a $\ket{0}$ state is indifferent to a subsequent $Z$ error). Thus, initialization and measurement fail with probabilities $p_i$ and $p_m$, respectively, by suffering from the bad Pauli error. All components succeed or fail independently. The ability to use a separate failure rate for each component is a nice benefit of exact counting. In principle, we need not even use isotropic, or even constant-in-time, depolarizing noise, but do so for simplicity. We also note that, for simplicity, we always assume recovery operations (even non-Pauli) are perfect.

Our ultimate goal is to calculate the probability of failure of an exREC, $P_{\text{fail}}=\text{Pr}\left[\text{fail}\right]$. However, doing so exactly would mean considering all combinations of faults, propagating them through the circuits, and checking the exREC correctness condition for each. Instead, as is typical for distance three codes, we count just up to two faults. Therefore, we can exactly calculate the quantities
\begin{align}
P^{(2)}_{\text{fail}}&=\text{Pr}\left[\text{fail},\le2\text{ faults}\right]\\
P^{(2)}_{\text{succ}}&=\text{Pr}\left[\neg\text{fail},\le2\text{ faults}\right],
\end{align}
which are functions of the depolarizing rates $p_a$ for $a\in\{\text{CCZ},\text{CNOT},H,I,i,m\}$. As a check, we make sure $P^{(2)}_{\text{fail}}+P^{(2)}_{\text{succ}}=1-O(p_a^3)$. These two-fault counts provide upper and lower bounds on $P_{\text{fail}}$ like
\begin{equation}
P^{(2)}_{\text{fail}}\le P_{\text{fail}}\le 1-P^{(2)}_{\text{succ}}.
\end{equation}
The upper and lower bounds calculated by our counting can be seen plotted in Figs.~\ref{ErrorRateBS} and \ref{ErrorRate10BS}. Solving $P_{\text{fail}}=p_g$ defines the pseudothreshold for a gate $g$. Solving instead $1-P_{\text{succ}}^{(2)}=p_g$ and $P^{(2)}_{\text{fail}}=p_g$ gives lower and upper bounds on the pseudothreshold.

\begin{figure}
\includegraphics[width=\columnwidth]{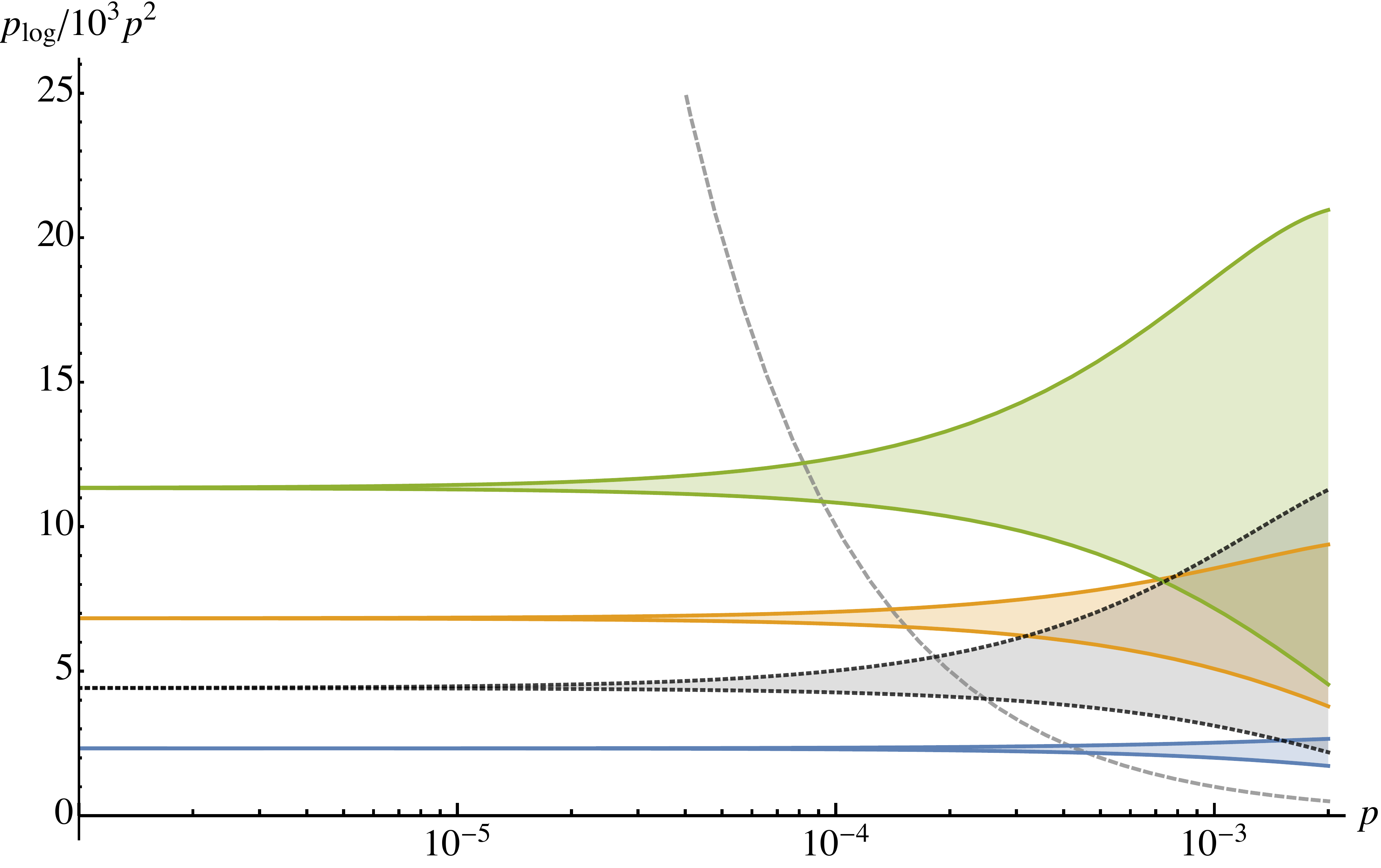}
\caption{\label{ErrorRateBS} Logical error rates $p_{\text{log}}$ of exRECs in the depolarizing circuit noise model in which all gates fail with probability $p$. We show from the bottom (1) the $3\times 3$ Bacon-Shor identity (or Hadamard) exREC (2) for comparison, the 9-qubit rotated surface code identity (as calculated in \cite{Yoder2017}) (3) the Bacon-Shor transversal CNOT (4) the Bacon-Shor CCZ. The dashed hyperbola $p_{\text{log}}=p$ is shown for estimating pseudothresholds.}
\end{figure}

\begin{figure}
\includegraphics[width=\columnwidth]{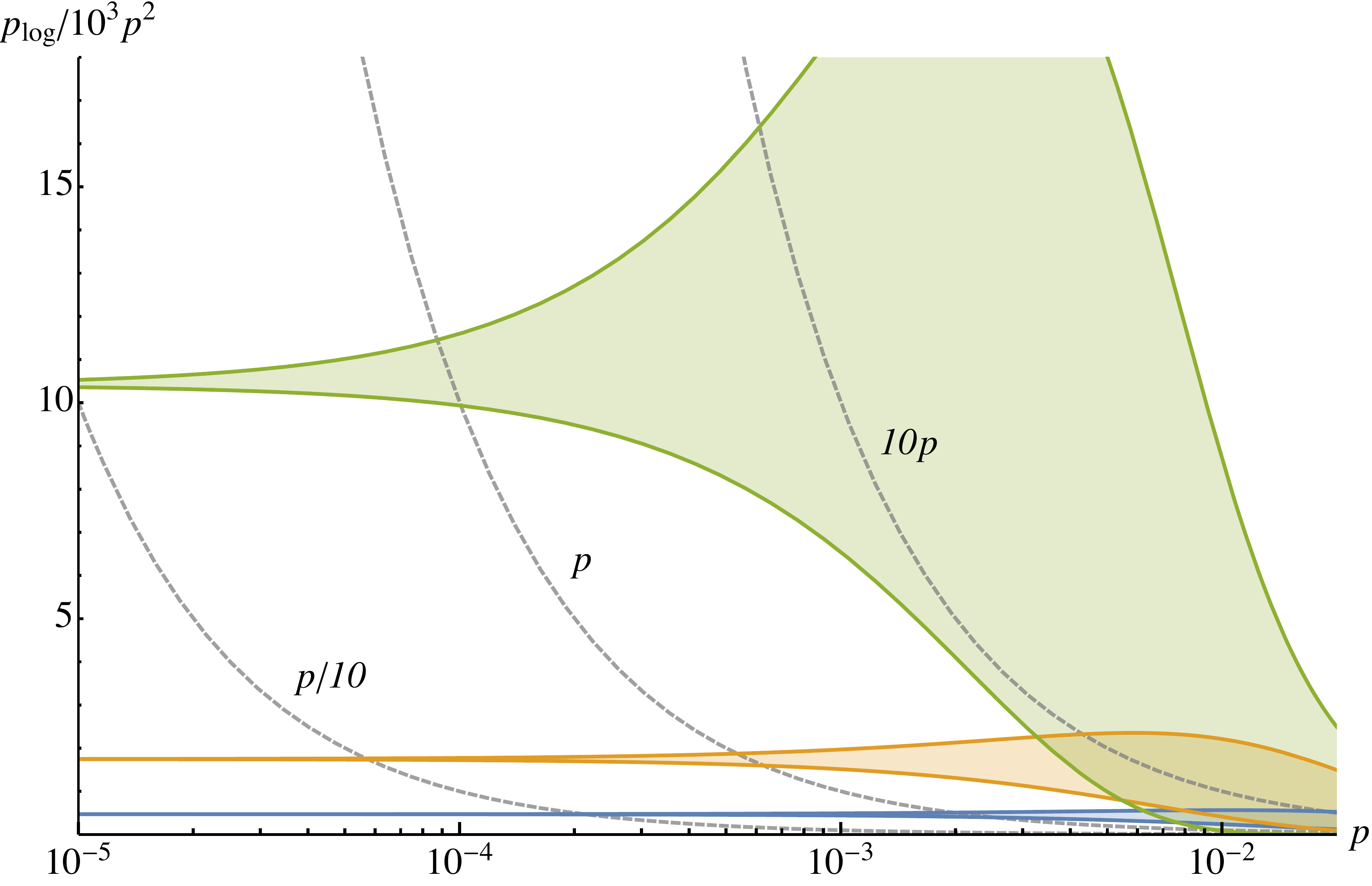}
\caption{\label{ErrorRate10BS} The same error rates for Bacon-Shor as Fig.~\ref{ErrorRateBS} but for circuit depolarizing noise that is harsher on larger components, $p_{\text{CCZ}}=10p$, $p_{\text{CNOT}}=p$, $p_{I}=p_{H}=p_i=p_m=p/10$. The hyperbolas at $p_{\text{log}}=p/10,p,10p$ are for estimating the pseudothreshold (of the CNOT failure probability) for the identity, CNOT, and CCZ exRECs, respectively. For CCZ, counting only up to two faults starts to become inaccurate.}
\end{figure}

Finally, we need to describe how errors, which begin as Pauli after the failed component, are tracked through our circuits. For exRECs consisting of Clifford gates (i.e.~all but the CCZ exREC) this tracking is a simple application of the Gottesman-Knill theorem \cite{Gottesman1998,Nielsen2000}. For the CCZ exREC, Pauli errors can become non-Pauli as they propagate through CCZ gates. Moreover, we cannot take a pessimistic approach and break these non-Pauli errors into several types of Pauli errors because our recovery works by explicitly correcting non-Pauli errors. Instead, we track all errors exactly, representing them as a sum of Pauli terms, storing the complex coefficients of each term in the sum. Stabilizer measurements break the Pauli sums in two, regrouping terms based on their commutation with the stabilizer measured. In our case, only measurement of the $X$-stabilizers actually does this, because terms in any given Pauli sum differ only in placement of $Z$s. To conclude this appendix, we argue this measurement mechanism is correct.

Our errors begin as unitaries (in fact, as Paulis) and are transformed by unitary conjugation as they progress. Thus, an error $E$ remains unitary. Since Paulis form an orthonormal basis under the Hilbert-Schmidt norm, we always have the ability to decompose $E$ into a \emph{Pauli sum},
\begin{equation}
E=\sum_\alpha a_\alpha \sigma_\alpha,
\end{equation}
where $a_\alpha\in\mathbb{C}$ and $\sigma_\alpha\in\mathcal{P}$ is in the Pauli group. Moreover, since $E^\dag E=I$ and thus $\text{Tr}[E^\dag E]=2^n$, we know $\sum_\alpha|a_\alpha|^2=1$. For two Paulis $\sigma,\gamma$, we define $[\sigma,\gamma]=0$ if they commute and $[\sigma,\gamma]=1$ if they anticommute.

Let us now assume that at the output of the circuit we are expecting the pure stabilizer state
\begin{equation}
\rho=\prod_{g\in G(S)}\frac12(I+g)=\frac{1}{2^n}\sum_{q\in S}q,
\end{equation}
where $S$ is an Abelian subgroup of $\mathcal{P}$ and $G(S)$ is a generating set for $S$. We actually have $E\rho E^\dag$ though. We can ask, what is the probability of finding one of the orthonormal states
\begin{equation}
\rho_{\vec m}=\prod_{g\in G(S)}\frac12(I+(-1)^{\vec m_g}g)
\end{equation}
at the output? This probability is
\begin{align}
\text{Pr}[\vec m]&=\text{Tr}\left[\rho_{\vec m}E\rho E^\dag\right]\\
&=\sum_\alpha a_\alpha\text{Tr}\left[\rho_{\vec m}\sigma_\alpha\rho E^\dag\right]\\
&=\sum_\alpha a_\alpha\text{Tr}\left[\sigma_\alpha\rho_{\vec m(\alpha)}\rho E^\dag\right],
\end{align}
where $\vec m(\alpha)_g=\vec m_g\oplus [\sigma_\alpha,g]$. Now, $\rho_{\vec m(\alpha)}\rho=0$ unless $\vec m(\alpha)=\vec 0$, in which case $\rho_{\vec m(\alpha)}=\rho$ and, because of purity, $\rho^2=\rho$. Thus,
\begin{align}
\text{Pr}[\vec m]&=\sum_{\alpha\text{ s.t.~}\vec m(\alpha)=\vec 0}a_\alpha\text{Tr}\left[\sigma_\alpha \rho E^\dag\right]\\
&=\sum_{\alpha\text{ s.t.~}\vec m(\alpha)=\vec 0}\sum_\beta a_\alpha a^*_\beta\text{Tr}\left[\sigma_\alpha\rho \sigma_\beta\right].
\end{align}
Now, if $\sigma_\beta \sigma_\alpha\in S$ (or, equivalently, $\vec m(\beta)=\vec m(\alpha)=\vec 0$), then the $(\alpha,\beta)$ trace term in the sum equals $1$. However, if $\sigma_\beta \sigma_\alpha\not\in S$, then the trace vanishes. We are left to conclude
\begin{equation}
\text{Pr}\left[\vec m\right]=|\sum_{\alpha\text{ s.t.~}\vec m(\alpha)=\vec0}a_\alpha|^2.
\end{equation}


\end{document}